\documentclass[a4paper,11pt]{amsart}
\usepackage{amsmath,amssymb,amsthm}
\usepackage[a4paper,margin=2cm]{geometry}
\usepackage{mathtools,stmaryrd}
\usepackage{url}
\usepackage{color}
\usepackage{enumitem}
\usepackage[all]{xy}
\usepackage{todonotes}

\usepackage{tikz}
\usetikzlibrary{arrows.meta, positioning, shapes.geometric}

\usepackage[colorlinks,linkcolor=blue, citecolor=blue]{hyperref}

\usepackage[style=alphabetic,url=false,doi=true,giveninits=true,isbn=false,maxbibnames=10]{biblatex}
\AtEveryBibitem{\clearlist{language}}
\AtEveryBibitem{\iffieldequalstr{entrykey}{Ng-Introduction}{\hypersetup{urlcolor=red}}{}}

\theoremstyle{plain}
\newtheorem{theorem}{Theorem}[section]

\newtheorem{corollary}[theorem]{Corollary}
\newtheorem{lemma}[theorem]{Lemma}

\newtheorem{proposition}[theorem]{Proposition}

\theoremstyle{definition}

\theoremstyle{remark}
\newtheorem{remark}[theorem]{Remark}

\newcommand{\Res}{\mathop{\mathrm{Res}}}

\title{Augmentation curve for torus knots}

\author{Alexander Hock}

\address{University of Geneva}

\email{Alexander.Hock@unige.ch}

\author{Sergey Shadrin}

\address{University of Amsterdam}

\email{S.Shadrin@uva.nl}

\begin{document}

\begin{abstract} We prove a version of topological recursion for the correlation differentials of colored HOMFLY-PT polynomials of torus knots on the corresponding augmentation curves, following the proposal of Gu, Jockers, Klemm, and Soroush. To this end, we explicitly construct the spectral curve from the so-called auxiliary curve of Brini, Eynard, and Mariño and prove that it coincides with the augmentation curve for every torus knot with coprime $(r,s)$. Finally, our version of topological recursion proves the
higher-genus closed-string prediction of the Aganagic--Vafa
conjecture for all torus knot augmentation curves.
\end{abstract}

\maketitle

\tableofcontents

\section{Introduction}

The colored HOMFLY-PT polynomials of torus $(r,s)$ knots are essentially combinatorial objects that can be computed with the Rosso-Jones~\cite{RossoJones} formula very explicitly. The Rosso-Jones formulas can be assembled into the \emph{Ooguri-Vafa partition function}~\cite{OoguriVafa}. Its cut-and-join representation was studied in~\cite{MirMorSle} leading to a free-fermion representation developed in~\cite{DPSS-HOMFLY}. In the free-fermion representation it is natural to consider the so-called \emph{extended Ooguri-Vafa partition function}~\cite{DPSS-HOMFLY}, which is a hypergeometric KP tau function. 

The data of both ordinary and extended Ooguri-Vafa partition functions can be efficiently stored in the so-called \emph{correlation differentials}, $\{\omega^{(g)}_{n}\}$ and $\{\omega^{\mathsf{ext},(g)}_{n}\}$, respectively; here $g\geq 0$ and $n\geq 1$. These systems of differentials are initially defined (at least, in the stable range $2g-2+n>0$ --- the cases $(g,n)=(0,1)$ and $(0,2)$ are a bit special) as expansions in formal power series. This raises a question whether these differentials, and hence the colored HOMFLY-PT polynomials of torus knots, can be controlled by \emph{topological recursion}~\cite{EynOra} on some underlying spectral curve. 

There are two answers to this question, one fully established and one conjectural, that involve the so-called \emph{auxiliary curve}~\cite{BEM11} and \emph{augmentation curve}~\cite{GJKS}, respectively. The goal of this paper is to turn the second answer proposed in~\cite{GJKS} as a conjecture into a mathematical theorem. To proceed with the explanation of the results of this paper, let us introduce the key objects.

\subsection{Colored HOMFLY-PT polynomial} The colored
HOMFLY-PT polynomial of the $(r,s)$ torus knot is given by the Rosso--Jones formula~\cite{RossoJones} in the so-called spectral  framing (cf.~\cite[Section 3]{DPSS-HOMFLY}) as
\begin{equation}\label{eq:Rosso-Jones}
	H_\lambda^{r,s}(Q^{-1/2},e^{\hbar/2})
	=Q^{-r|\lambda|/2}
	\sum_{\mu\vdash s|\lambda|}c_\lambda^\mu
	e^{\frac rs\hbar\sum_{(i,j)\in\mu}(i-j)}
	s_\mu\left(\frac{\hat q_1}{\hbar},\frac{\hat q_2}{\hbar},\ldots\right),
\end{equation}
where $(i,j)\in \mu$ denotes a box of the Young diagram associated with the partition $\mu$ with column index $i$ and row index $j$, defining the content of a box by $i-j$,
\begin{align}
\frac{\hat q_k}{\hbar}
	=\frac{Q^{-k/2}-Q^{k/2}}{e^{k\hbar/2}-e^{-k\hbar/2}},\ k\geq 1,    
\end{align}
and $c_\lambda^\mu$ are the Adams coefficients defined in terms of the Schur functions $s_\lambda$ as 
\begin{align}
	s_\lambda(p_s,p_{2s},\ldots)
	=\sum_{\mu\vdash s|\lambda|}c_\lambda^\mu s_\mu(p).
\end{align}
The Ooguri-Vafa partition function is defined as:
\begin{align}\label{eq:OV}
	Z(p)
	=\sum_\lambda Q^{r|\lambda|/2}
	H_\lambda^{r,s}(Q^{-1/2},e^{\hbar/2})s_\lambda(p).
\end{align}
The differentials $\omega^{(g)}_n$, $g\geq 0$, $n\geq 1$, that retain the information about the HOMFLY-PT polynomials are defined as formal power series in the variables $\tilde{X}_1,\dots,\tilde{X}_n$:
\begin{equation}\label{eq:omega-g-n}
	\omega^{(g)}_n
	=\sum_{m_1,\ldots,m_n\geq1}
	\left.[\hbar^{2g-2+n}]
	\frac{\partial^n\log Z(p)}
	{\partial p_{m_1}\cdots\partial p_{m_n}}\right|_{p=0}
	\prod_{i=1}^n m_i\tilde X_i^{m_i}
	\frac{d\tilde X_i}{\tilde X_i} + \delta_{(g,n)=(0,2)}\frac{d\tilde X_1 d\tilde X_2}{(\tilde{X}_1-\tilde{X}_2)^2}.
\end{equation}
Here and below by $[\hbar^a]$ we mean the operator of extraction of the coefficient of $\hbar^a$. 

\subsection{Fock space formalism for the correlation differentials} 

We use the charge $0$ bosonic Fock space with Schur basis $v_\lambda$ and
\begin{align}
	[J_k,J_l]=k\delta_{k+l,0},\qquad
	J_k|0\rangle=0,\qquad \langle0|J_{-k}=0\quad(k>0).
\end{align}
Define the following two operators, both Laurent series in $\hbar$ (more precisely, the second operator is expanded with respect to the energy grading with coefficients being Laurent series in $\hbar$):
\begin{align}\label{eq:D}
	\mathcal {D}\colon v_\lambda \mapsto 
	&\exp\Big(\frac rs\hbar\sum_{(i,j)\in\lambda}(i-j)\Big)v_\lambda, \qquad 
	\mathcal{J} \coloneqq \exp\Big(\sum_{k\geq1}\frac{\hat q_kJ_{-k}}{k\hbar}\Big).
\end{align}
The Ooguri-Vafa partition function can equivalently be written as
\begin{align}
    Z(p)=\big\langle0\big|
	\exp\bigg(\sum_{m\geq 1}\frac{J_{sm}p_m}{m}\bigg)\mathcal{DJ}
	\big|0\big\rangle,
\end{align}
giving rise to the definition of the disconnected $n$-point function $W^\bullet_n$ as the following vacuum expectation value expression:
\begin{align}\label{eq:correlator W}
	W^\bullet_n(\tilde X_1,\ldots,\tilde X_n)
	\coloneqq \sum_{m_1,\ldots,m_n\geq1}
	\Big(\prod_{i=1}^n\tilde X_i^{m_i}\Big)
	\big\langle0\big|
	\Big(\prod_{i=1}^nJ_{s m_i}\Big)\mathcal{DJ}
	\big|0\big\rangle.
\end{align}
The connected $n$-point function is obtained by the inverse moment-cumulant formula:
\begin{equation}\label{eq:connected-from-disconnected}
	W_n(\tilde X_1,\dots,\tilde X_n)
	=\sum_{l=1}^n\frac{(-1)^{l-1}}{l}
	\sum_{\substack{I_1\sqcup\cdots\sqcup I_l=\llbracket n\rrbracket\\
			I_j\ne\emptyset}}
	\prod_{j=1}^l W^\bullet_{|I_j|}(\tilde X_{I_j}).
\end{equation}
Here the inner sum goes over ordered collections of nonempty disjoint sets.
Equivalently, one can sum over unordered set partitions with coefficient
$(-1)^{l-1}(l-1)!$. Note that $W_n$ are formal power series in $\tilde X_i$, $i=1,\dots, n$, $n\geq 1$. For $g\geq 0$, $n\geq 1$, we have
\begin{equation}\label{eq:omega-Fock}
	\omega^{(g)}_n
	\coloneqq[\hbar^{2g-2+n}]W_n
	\prod_{i=1}^n\frac{d\tilde X_i}{\tilde X_i} + \delta_{(g,n)=(0,2)}\frac{d\tilde X_1 d\tilde X_2}{(\tilde{X}_1-\tilde{X}_2)^2}.
\end{equation}

\subsection{Extended correlation differentials} The extended Ooguri-Vafa partition function is in terms of Fock space vacuum expectation value given by
\begin{align}
    Z^{\mathsf{ext}}(p)=\big\langle0\big|
	\exp\bigg(\sum_{m\geq 1}\frac{J_{m}p_m}{m}\bigg)\mathcal{DJ}
	\big|0\big\rangle
\end{align}
without $s$ in $J_{m}$ of the exponential. The extended correlation differentials are defined as 
\begin{equation}
	\omega^{\mathsf{ext},(g)}_n
	\coloneqq[\hbar^{2g-2+n}]W^{\mathsf{ext}}_n
	\prod_{i=1}^n\frac{dX_i}{X_i} + \delta_{(g,n)=(0,2)}\frac{dX_1 dX_2}{({X}_1-{X}_2)^2},
\end{equation}
where $W^{\mathsf{ext}}_n(X_1,\dots,X_n)$ are formal power series obtained by the inverse moment-cumulant formula from
\begin{align}\label{eq:extendedcorrelator W}
	W^{\mathsf{ext},\bullet}_n(X_1,\dots,X_n)
	\coloneqq \sum_{m_1,\ldots,m_n\geq1}
	\Big(\prod_{i=1}^n X_i^{m_i}\Big)
	\big\langle0\big|
	\Big(\prod_{i=1}^nJ_{m_i}\Big)\mathcal{DJ}
	\big|0\big\rangle
\end{align}
which itself is derived from the extended partition function. So equivalently, the extended correlation differentials can be written as
\begin{equation}
	\omega^{\mathsf{ext},(g)}_n
	=\sum_{m_1,\ldots,m_n\geq1}
	\left.[\hbar^{2g-2+n}]
	\frac{\partial^n\log Z^{\mathsf{ext}}(p)}
	{\partial p_{m_1}\cdots\partial p_{m_n}}\right|_{p=0}
	\prod_{i=1}^n m_i X_i^{m_i}
	\frac{d X_i}{ X_i} + \delta_{(g,n)=(0,2)}\frac{d X_1 d X_2}{(X_1-X_2)^2}.
\end{equation}

The extended correlation differentials satisfy the topological recursion~\cite{EynOra} on $\mathbb{C}\mathrm{P}^1$ with the standard Bergman kernel $B(z_1,z_2) = \frac{dz_1 dz_2}{(z_1-z_2)^2}$ in the global affine coordinate $z$, with the functions $X(z) = \exp(x(z))$ and  $y(z)$ given by
\begin{align}
\label{eq:X}
X(z) & = z\exp\left(-\frac rs \left( \log \left(1-Q^{\frac 12} z \right) -\log \left(1-Q^{-\frac 12} z \right) \right)\right) ; 
\\  \label{eq:y}
y(z) &= \log \left(1-Q^{\frac 12} z \right) -\log \left(1-Q^{-\frac 12} z \right).
\end{align}
This means the following properties (see e.g.~\cite{BEO,BS}):
	\begin{enumerate}
		\item The differential $\omega^{\mathsf{ext},(0)}_1(X)$ is the expansion in $X$ used as a local coordinate at $z=0$ of $y(z)dx(z)$.  In the same sense, $\omega^{\mathsf{ext},(0)}_2(X_1,X_2)$ is the expansion of $B(z_1,z_2)$, and for  $2g-2+n>0$ the differentials $\omega^{\mathsf{ext},(g)}_n$ are expansions of globally defined meromorphic $n$-differentials on $(\mathbb{C}\mathrm{P}^1)^n$. In all these cases we denote the corresponding globally defined differentials still by $\omega^{\mathsf{ext},(g)}_n$, abusing notation. 
		\item For $2g-2+n>0$ the only poles of $\omega^{\mathsf{ext},(g)}_n$ (considered as global meromorphic $n$-differentials) are the points $p_1$ and $p_2$ satisfying $dx|_{p_i} = 0$, $i=1,2$.
		\item The system of differentials $\{\omega^{\mathsf{ext},(g)}_{n}\}$ satisfies the so-called \emph{linear and quadratic loop equations} at $p_1$ and $p_2$, which determine their principal parts there.
		\item For $2g-2+n>0$, the differentials $\omega^{\mathsf{ext},(g)}_{n}$ satisfy the projection property, that is, they can be reconstructed from their principal parts at $p_1,p_2$ by
		\begin{align}
			\omega^{\mathsf{ext},(g)}_{n}(z_1,\dots,z_n) = \sum_{a_1,\dots,a_n=1,2} \Big( \prod_{i=1}^n \Res_{z'_i\to p_{a_i}}	\int_{p_{a_i}}^{z'_i} B(\bullet,z_i) \Big) \omega^{\mathsf{ext},(g)}_{n}(z'_1,\dots,z'_n).	
		\end{align}
	\end{enumerate}

\begin{remark}
	$Q^{1/2}$ can be either treated as a given constant in general position with respect to all computations, or we can work over the field of coefficients $\mathbb{C}(Q^{1/2})$. The latter choice is the most convenient one for this paper. 
\end{remark}

\begin{remark} For completeness, let us recall the loop equations. Let $p$ be a branch point of $X$ and $t$ the local coordinate at $p$ such that $X(t) = X(p) + t^2$. Then the linear loop equation reads
\begin{align} \label{eq:linear-loop-ext}
    \omega^{\mathsf{ext},(g)}_{n+1}(t,z_1,\dots,z_n) + \omega^{\mathsf{ext},(g)}_{n+1}(-t,z_1,\dots,z_n)
\end{align}
is holomorphic with at least a simple zero at $t=0$, $g,n\geq 0$. The quadratic loop equation reads 
\begin{align} \label{eq:quadratic-loop-ext}
    \omega^{\mathsf{ext},(g-1)}_{n+2} (t,-t,z_1,\dots,z_n) + \sum_{\substack{g_1+g_2 = g\\ I_1\sqcup I_2 = \{1,\dots,n\} }}\omega^{\mathsf{ext},(g_1)}_{|I_1|+1}(t,\{z_i\}_{i\in I_1}) \omega^{\mathsf{ext},(g_2)}_{|I_2|+1}(-t,\{z_i\}_{i\in I_2})
\end{align}
is holomorphic with at least a double zero at $t=0$, $g,n\geq 0$. Here we assume the convention that  $\omega^{\mathsf{ext},(-1)}_{n+2}=0$, $n\geq 0$.
\end{remark}

This spectral curve data is exactly the auxiliary curve mentioned above. It was proposed in~\cite{BEM11} as a tool to control the differentials $\omega^{(g)}_n$, in the context of matrix models. In application to the extended differentials $\omega^{\mathsf{ext},(g)}_{n}$ it was first considered in~\cite{DPSS-HOMFLY}, where the points (1), (2), and (4) as well as the linear loop equations were proved. Subsequently, the quadratic loop equations were proved in~\cite{DKPSS-HOMFLY}, and then all these statements got one more completely different proof in~\cite{BDKS-TR}. One more recent and yet another proof of all these statements (1-4) is available in~\cite{ABDKS-Symplectic}, where it is based on the theory of logarithmic topological recursion initiated in~\cite{Hock-qDiLog} and developed in~\cite{ABDKS-LogTR}.

\subsection{From extended to the ordinary differentials} Let $\xi$ be a primitive $s$-th root of unity. For $(g,n)\not=(0,2)$ we have
\begin{align} \label{eq: averaging}
	\omega^{(g)}_{n} (\tilde X_1,\dots,\tilde X_n) = \sum_{\alpha_1,\dots,\alpha_n=1}^s 	
	\omega^{\mathsf{ext},(g)}_{n} ( \xi^{\alpha_1}X_1,\dots,\xi^{\alpha_n}X_n),
\end{align}
as formal power series in $\tilde X = X^s$. In the case $(g,n)=(0,2)$ some adjustment of the term with the double pole on the diagonal is required. 

This connection allowed in~\cite{BEM11} to make a proposal on how to control the ordinary differentials 	$\omega^{(g)}_{n}$ in terms of the same auxiliary spectral curve data, and in~\cite{BEO} this proposal was proved. However, according to~\cite{GJKS}, it is not the final form of the desired answer in this case. 

In~\cite{JKS} and~\cite{GJKS} the ordinary differentials were studied further, and it was suggested that Equation~\eqref{eq: averaging} can be used to revisit the topological recursion statement for the ordinary differentials on the true physical curve defined using the $\omega^{(0)}_{1}$ rather than $\omega^{\mathsf{ext},(0)}_{1}$. The degree bounds for this curve were conjectured in~\cite{JKS}, and then it was argued in~\cite{GJKS} that this true physical curve must be identified as the augmentation polynomial of the corresponding torus knot~\cite{Ng-definition,Ng-Introduction}. There are deep physical reasons within the Chern-Simons theory to believe that it is true (both the statement on the origin of the true physical spectral curve data and the topological recursion statement), and in any case this proposal was tested on a number of examples in~\cite{GJKS}. 

\subsection{What this paper establishes}
We will established mathematical proofs of conjectures raised in \cite{AV12,JKS,GJKS} regarding the torus knots.

In Section~\ref{sec:TrueCurve} we recover the true physical curve very explicitly, proving the degree bounds conjectured in~\cite{JKS}. We further show that the corresponding quantum curve can be effectively derived from the known quantum curve of the auxiliary curve. In Section~\ref{sec:Augmentation} we follow the exposition of the construction of the augmentation polynomial/curve for the case of torus knots extracted from~\cite{Ng-Introduction} (very conveniently this case was featured there as exercises for the reader), and prove that it coincides with the physical curve we obtained in Section~\ref{sec:TrueCurve}.  This confirms a conjecture in~\cite{GJKS}. 

Finally, in Section~\ref{sec: topological recursion} we use the results on topological recursion for the extended differentials on the auxiliary curve to derive the topological recursion statement on the true physical curve $=$ augmentation curve, confirming the conjectural proposal of~\cite{GJKS}. 
This section is the least original of this paper, since it adapts to the situation that we have here  (the new understanding of the true physical curve and its interaction with the auxiliary curve; including prior results on topological recursion for the extended differentials on the auxiliary curve) some arguments on averaging of the loop equations already presented in earlier works, notably in~\cite{BEO}. We also show that the differentials $\omega_n^{(g)}$ on the
augmentation curve are invariant under interchanging $r$ and $s$,
up to a rescaling of $\tilde X$, consistently with their
interpretation as knot invariants. Finally, we prove a normalized
dilaton equation for these differentials and show that the
corresponding free energies $\omega_0^{(g)}$ coincide with the
auxiliary free energies $\omega_0^{\mathsf{ext},(g)}$.
This yields the resolved conifold free energy formula for $g\geq2$ for all
torus knots on their corresponding augmentation curves, in agreement with the closed-string
prediction of Aganagic and Vafa \cite{AV12}.

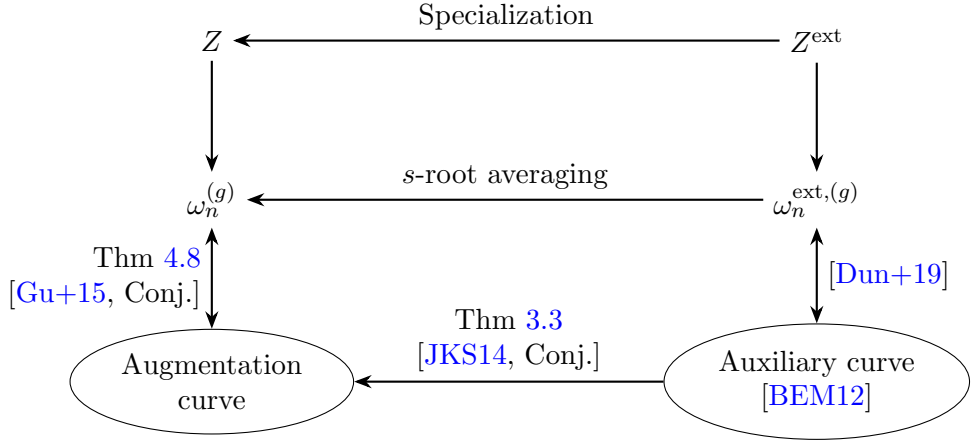
\begin{figure}[htbp]
\centering
\begin{tikzpicture}[
    >=Stealth,
    node distance=2cm,
    every node/.style={font=\normalsize},
    curve/.style={
        draw,
        ellipse,
        align=center,
        minimum width=3.2cm,
        minimum height=1.2cm
    },
    arrow/.style={->, thick}
]

\node (Z) at (0,4.5) {$Z$};
\node (Zext) at (8,4.5) {$Z^{\mathrm{ext}}$};

\node (omega) at (0,2.4) {$\omega_n^{(g)}$};
\node (omegaext) at (8,2.4)
    {$\omega_n^{\mathrm{ext},(g)}$};

\node[curve] (aug) at (0,0)
    {Augmentation\\curve};
\node[curve] (aux) at (8,0)
    {Auxiliary curve\\ \cite{BEM11}};

\draw[arrow] (Zext) --
    node[above] {Specialization}
    (Z);

\draw[arrow] (Z) -- (omega);
\draw[arrow] (Zext) -- (omegaext);

\draw[arrow] (omegaext) --
    node[above] {$s$-root averaging}
    (omega);

\draw[<->, thick] (aux) --
    node[right, align=left] {\cite{DPSS-HOMFLY}}
    (omegaext);

\draw[<->, thick] (omega) --
    node[left, align=right] {
        Thm~\ref{thm:TRonAugmentation}\\
         \cite[Conj.]{GJKS}
    }
    (aug);

\draw[arrow] (aux) --
    node[above, align=center] {
        Thm~\ref{thm:FF}\\
         \cite[Conj.]{JKS}
    }
    (aug);

\end{tikzpicture}
\caption{Relations between the (extended) Ooguri--Vafa partition
function, (extended) correlation differentials, and the auxiliary
and augmentation curves for torus knots.}
\label{fig:OV-augmentation-diagram}
\end{figure}

\subsection{Acknowledgements} This project was prompted by questions posed to us by Vincent Bouchard and Piotr Su{\l}kowski at the program ``Enumerative Geometry, Integrability, and Topological Recursion'' organized at the Matrix Center in Creswick, Australia. We thank them for very useful discussions and we thank the organizers of the program, Alexander Alexandrov,
Norman Do, and Paul Norbury for inviting us to this very stimulating research environment. A.H. thanks Marcos Mariño for discussions.

A.H. was funded by the Swiss National Science Foundation (SNSF) through the
Ambizione project “TRuality: Topological Recursion, Duality and Applications”
under the grant agreement PZ00P2\_\allowbreak223297 and partially by the DFG -- project ID 551478549.
S.S. was supported by the Dutch Research Council, grant no. OCENW.M.21.233.

AI agents were used substantially in producing and polishing the proofs in this paper. They were fed with suggestions of  algebraic tricks and ideas of proofs and asked to list consequences, and were very instrumental in formal computations of this type. However, the final versions of all proofs and the composition and text of the paper are still produced manually.

\section{True curve for torus knots} \label{sec:TrueCurve}

\subsection{Derivation of the polynomial equation} Let $r,s$ be positive integers. Recall the data of the auxiliary spectral curve given in Equations~\eqref{eq:X} and~\eqref{eq:y} and recall that we work over $\mathbb{C}(Q^{\frac 12})$. 
$X$ can be used as a local coordinate at $z=0$, the change $X=X(z)$ is invertible, we can write $z=z(X)$. Consider the series $y(z)$ in the coordinate $X$, that is, substitute $y(z(X))$. In the resulting series select only the terms whose powers of $X$ are divisible by $s$, and write $\tilde X\coloneqq X^s$, that is, consider the series $\tilde y(\tilde X)$ defined by 
\begin{align} \label{eq: tilde y}
	\tilde y(\tilde X)|_{\tilde X = X^s} = \frac 1s \sum_{i=0}^{s-1} y(z(\xi^i X)),
\end{align}
where $\xi$ is a fixed primitive $s$-th root of unity. Let $\tilde Y = \exp(\tilde y)$. 

\begin{theorem}\label{thm:corrected-algebraicity}
Assume $\gcd(r,s)=1$. There is a nonzero algebraic relation
\begin{equation}\label{eq: True Spectral Curve}
 \sum_{i,j}c_{i,j}\tilde X^i\tilde Y^j=0,
\end{equation}
{where the coefficients $c_{i,j}$ belong to $\mathbb C(Q^{1/2})$ and are not all zero.}
In particular, the ranges of the exponents are
\begin{align}
 0\leq i\leq\frac1{r+s}\binom{r+s}{r},\qquad
 0\leq j\leq\binom{r+s}{r}.
\end{align}
The relation can be chosen to have these exact degrees
in $\tilde X$ and $\tilde Y$ (note that for coprime $r$ and $s$, $\frac1{r+s}\binom{r+s}{r}$ is an integer), with $c_{\frac1{r+s}\binom{r+s}{r},\binom{r+s}{r}} = 1$.
\end{theorem}

\begin{remark} This theorem proves a conjecture stated in~\cite{JKS}, see Eq.~(4.41) in \emph{op.~cit.}
\end{remark}

\begin{proof}
	First, we construct a polynomial and then we prove that it indeed gives a relation for $\tilde X$ and $\tilde Y$. From the definition of $X(z)$ we have
	\begin{align}\label{eq:tildeX}
	\tilde X=X(z)^s
	=z^s\left(\frac{1-Q^{-1/2}z}{1-Q^{1/2}z}\right)^r.
	\end{align}
	Thus, for a generic value of $\tilde X$, there are $r+s$ distinct roots $z_1,\ldots,z_{r+s}$ of $P(z)=0$ where
	\begin{equation}\label{eq:z-root-polynomial}
		P(z)\coloneqq z^s(1-Q^{-1/2}z)^r-\tilde X(1-Q^{1/2}z)^r.
	\end{equation}
	Here $z_1,\ldots,z_{r+s}$ denote all the algebraic inverse branches
	as functions of $\tilde X$.
	
	Choose integers $p,q$ such that $rp-sq=1$. For each subset
	$S\subset\{1,\ldots,r+s\}$ with $|S|=s$, define the algebraic function
	\begin{equation}\label{eq:cS-definition}
		c_S=
		\left((-1)^{s-1}\frac{\prod_{i\in S}z_i}{\tilde X}\right)^p
		\left(\prod_{i\in S}
		\frac{1-Q^{-1/2}z_i}{1-Q^{1/2}z_i}\right)^{q}.
	\end{equation}
	Though $p$ and $q$ enter the definition of $c_S$, it is easy to see that $c_S$ doesn't depend on these choices. Indeed, multiplying the
	identities \eqref{eq:z-root-polynomial} for $i\in S$ gives
	\begin{align}
	\left((-1)^{s-1}\frac{\prod_{i\in S}z_i}{\tilde X}\right)^{-s}
	=\left(\prod_{i\in S}
	\frac{1-Q^{-1/2}z_i}{1-Q^{1/2}z_i}\right)^r,
	\end{align}
	where the sign doesn't destroy the identity because $s(s-1)$ is even. Together with $rp-sq=1$, this implies
	\begin{align}\label{eq:cS-powers}
	c_S^r
		=(-1)^{s-1}\frac{\prod_{i\in S}z_i}{\tilde X},
		\qquad 
	c_S^s
		=\prod_{i\in S}\frac{1-Q^{1/2}z_i}{1-Q^{-1/2}z_i}.
	\end{align}
	For coprime $(r,s)$ these equations imply that $c_S$ is independent of the
	choice of $p,q$.
	
	Now we define the equation. We have:
	\begin{align}\label{eq:F}
		F(\tilde X,\tilde Y) & =
		\tilde X^{\frac1{r+s}\binom{r+s}{s}}
		\prod_{\substack{S\subset\{1,\ldots,r+s\}\\ |S|=s}}
		(\tilde Y-c_S)
		\\ 
		\notag
		& = \sum_{k=0}^{\binom{r+s}{s}} \tilde Y^{\binom{r+s}{s}-k }\tilde X^{\frac1{r+s}\binom{r+s}{s}} (-1)^k e_k(\{c_S\}_{|S|=s}). 
	\end{align}
	Note that $e_k(\{c_S\}_{|S|=s})$ is a symmetric expression in $z_1,\dots,z_{r+s}$ hence a rational function in $\tilde X$ (and $Q^{\frac 12}$). Let's show that it is a Laurent polynomial with the order of pole at $0$ at most $\frac1{r+s}\binom{r+s}{s}$ and bounded at infinity (hence, $\tilde X^{\frac1{r+s}\binom{r+s}{s}} (-1)^k e_k(\{c_S\}_{|S|=s})$ is a polynomial of degree at most $\frac1{r+s}\binom{r+s}{s}$).

Note that for $\tilde X\ne0$ none of the roots of \eqref{eq:z-root-polynomial} is equal to $0,Q^{\pm1/2}$.
The leading coefficient of $P$ is the nonzero constant $(-Q^{-1/2})^r$, so its roots remain finite over every finite value of $\tilde X$.
Hence the expression for $c_S$ shows that $e_k(\{c_S\}_{|S|=s})$ can have poles only at $\tilde X=0,\infty$, and is therefore a Laurent polynomial.
	Then, let us estimate the order of its poles at $\tilde X=0,\infty$. 
	
	 At $\tilde X=0$ the roots $z_1,\dots,z_{s+r}$ split into $s$ roots near $0$ with $z_i\asymp\tilde X^{1/s}$ 
	 and $r$ roots near $Q^{1/2}$ with $z_i-Q^{1/2}\asymp\tilde X^{1/r}$. Hence, the order of each $c_S$ is non-positive (since the order of $c_S^r$ is at most $0$). Hence, the order of pole of $e_k(\{c_S\}_{|S|=s}) $ is at most the order of pole of $e_{\binom{r+s}{s}}(\{c_S\}_{|S|=s})$ that is computed by
	 \begin{align}
	 	e_{\binom{r+s}{s}}(\{c_S\}_{|S|=s}) \asymp\tilde{X}^{\binom{r+s}{s} \big( \frac{s^2p}{s(s+r)} + \frac{srq}{r(s+r)}-p \big)} 
	 	= \tilde{X}^{-\frac{1}{s+r}\binom{r+s}{s}}.
	 \end{align}
	At $\tilde X=\infty$, $r$ roots tend to $Q^{-1/2}$
	and $s$ tend to infinity. For the former $r$ roots,
	\begin{align}
	\frac{1-Q^{1/2}z_i}{1-Q^{-1/2}z_i}
	\asymp\tilde X^{-1/r};
	\end{align}
	and for the latter $s$ roots this fraction tends to $Q$. 
Thus, $c_S^s$ and hence $c_S$ are bounded at $\tilde X=\infty$. Thus, $e_k(\{c_S\}_{|S|=s})$ is bounded at $\tilde X=\infty$ as well.

Thus, the only possible pole of $e_k(\{c_S\}_{|S|=s})$ is at $\tilde X = 0$, and the order of this pole is $\leq \frac{1}{s+r}\binom{r+s}{s}$. Hence, $\tilde X^{\frac1{r+s}\binom{r+s}{s}} (-1)^k e_k(\{c_S\}_{|S|=s})$ is a polynomial of degree  at most ${\frac1{r+s}\binom{r+s}{s}} $.
	
	It remains to show that $F$ vanishes on the given $\tilde{X}$ and $\tilde{Y}$. 
Let $S_0\subset \{1,\dots,r+s\}$ be the subset that indexes the roots $z_i$ such that $z_i\to 0$ for $\tilde{X}\to 0$. In other words, these are exactly the choices $z(\xi^iX)$, $i=0,\dots,s-1$, of Equation~\eqref{eq: tilde y}. Let us show that then we have $\tilde{Y} = c_{S_0}$. Indeed, from ${X}(z) = z \exp(-\frac rs y(z))$ we have:
	\begin{align}
		\xi^iX
		=z(\xi^iX)\,e^{-\frac rs y(z(\xi^iX))}, \qquad z(\xi^iX)=\xi^iX\,e^{\frac rs y(z(\xi^iX))}.
	\end{align}
	Hence, by Equation~\eqref{eq: tilde y},
	\begin{align}
	& (-1)^{s-1}\frac{\prod_{i\in S_0}z_i}{\tilde X}
	= \frac {(-1)^{s-1}}{\tilde X} \Big(\prod_{i=0}^{s-1}\xi^i\Big) X^s e^{\frac rs \sum_{i=0}^{s-1} y(z(\xi^iX))} = \tilde Y^r,
	\\
	& \prod_{i\in S_0}
	\frac{1-Q^{-1/2}z_i}
	{1-Q^{1/2}z_i}
	=\exp\Big(-\sum_{i=0}^{s-1}y(z(\xi^iX))\Big)
	=\tilde Y^{-s}.
	\end{align}
	Hence $c_{S_0}=\tilde Y^{rp-sq}=\tilde Y$, thus one of the factors in \eqref{eq:F} does vanish on given $\tilde{X},\tilde{Y}$. 
Finally, notice that the coefficient of $\tilde{Y}^{\binom{r+s}s}$ is precisely $\tilde{X}^{\frac 1{r+s}\binom{r+s}s}$, so we got exactly the form of the relation as stated.
\end{proof}

\subsection{Properties of \texorpdfstring{$F$}{F}} Note that Eq.~\eqref{eq:F} is very explicit. Various coefficients and specializations can be extracted without much effort. For instance,
\begin{lemma} \label{lem: constant term} The constant term of $F(\tilde{X},\tilde{Y})$ as a polynomial in $\tilde{Y}$ is equal to $({Q^{1/2}})^{\frac{s}{r+s}\binom{r+s}{s}}$.
\end{lemma}

\begin{proof} The factorization of Equation~\eqref{eq:z-root-polynomial} is 
\begin{align}
    P(z)=(-Q^{-1/2})^{r}\prod_{i=1}^{r+s}(z-z_i),
\end{align}
which
implies at different evaluation points that
\begin{align} \label{eq:total-products}
	\prod_{i=1}^{r+s}z_i&=(-1)^{s-1}Q^{r/2}\tilde{X};
	& \prod_{i=1}^{r+s}\frac{1-Q^{-1/2}z_i}{1-Q^{1/2}z_i}
	&=(-1)^{r-1}Q^{-s/2}\tilde{X}.
\end{align}
Hence,
	\begin{align}
		(-1)^{\binom{r+s}{s}}
		\tilde X^{\frac1{r+s}\binom{r+s}{s}}
		e_{\binom{r+s}{s}}(\{c_S\}_{|S|=s})
		=({Q^{1/2}})^{\frac{s}{r+s}\binom{r+s}{s}}.
	\end{align}
\end{proof}

As the first structural property, we prove the following:

\begin{proposition} \label{prop:irreducible} The polynomial $F$ is irreducible in $\mathbb{C}(Q^{1/2})[\tilde{X},\tilde{Y}]$. 
\end{proposition}

\begin{proof} 
To prove irreducibility over $\mathbb C(Q^{1/2})(\tilde X)$, we consider the monodromy of $\tilde{X}$ understood as rational function on the auxiliary curve, that is, the covering $\tilde{X}\colon \mathbb{C}\mathrm{P}^1 \to  \mathbb{C}\mathrm{P}^1$, $\tilde{X}(z)=X(z)^s$. It is sufficient to show that
\begin{itemize} 
\item The monodromy group is the full symmetric group $S_{r+s}$, hence it acts transitively on subsets of a generic fiber of cardinality $s$. 
\item 
On a generic fiber all $c_S$ are different, hence the transitive action on $s$-element subsets of the fiber gives a transitive action on $\{c_S\}_{|S|=s}$. 
\end{itemize} 

Since the covering is connected, the monodromy is transitive. It is generated by four permutations: two transpositions $\tau_1$ and $\tau_2$ and two permutations $\sigma_0$ and $\sigma_\infty$, each consisting of two disjoint cycles of length $r$ and $s$. The product of $\sigma_0$  with one of the transpositions, say $\tau_1$, must be a cycle of length $r+s$. Without loss of generality assume that  $\sigma_0 = \mathit{(1,\dots,s)(s+1,\dots,s+r)}$ and $\tau_1 = \mathit{(1,s+1)}$.  Since $r$ and $s$ are coprime, conjugating $\tau_1$ with powers of $\sigma_0$ we obtain any transposition $\mathit{(i,j)}$ for $i=1,\dots,s$, $j=s+1,\dots,s+r$. These transpositions generate $S_{r+s}$, hence, it is the monodromy group. 

Now, assume that $c_S=c_{S'}$ as algebraic functions, for $S\ne S'$.
Equation~\eqref{eq:cS-powers}, after cancelling the common factors,
gives
$\prod_{i\in S\setminus S'}z_i
=\prod_{i\in S'\setminus S}z_i$.
If we are not in the case $r=s=1$, then we can choose $l\in S\setminus S'$ and $k\not\in (S\setminus S')\cup (S'\setminus S)$ outside
this union. The  transposition $\mathit{(k,l)}$ from the monodromy group preserves the identity, hence $z_k=z_l$,
which contradicts distinctness of the generic roots.
In the only remaining case $r=s=1$
Equation~\eqref{eq:cS-powers} gives $c_{\{i\}}=z_i/\tilde X$, so the
two values are distinct as well.

Finally, there is still a possibility that there is a factor in $F$ that does not depend on $\tilde Y$. This is possible only for a positive integral exponent of $\tilde X$, if the maximal order of poles of $e_k(\{c_S\}_{|S|=s})$ is less than $\frac{1}{r+s}\binom{r+s}{s}$. However, since by Lemma~\ref{lem: constant term} $F(\tilde X,0)=({Q^{1/2}})^{\frac{s}{r+s}\binom{r+s}{s}}\ne0$, no positive power of $\tilde X$ divides all coefficients of $F$. 
\end{proof}

\subsection{Properties of the curve}  
Let $C$ be the smooth projective completion of the normalization of the affine curve defined by $F(\tilde X,\tilde Y)=0$.
The monodromy analysis in the proof of Proposition~\ref{prop:irreducible} can be expanded further to describe all ramification points of $\tilde{X}$ considered as a ramified covering $C\to \mathbb{C}\mathrm{P}^1$. 

Recall that $dx=d\log X$ has two zeros on the auxiliary curve, $p_1$ and $p_2$. Let $V_1 = ({X}(p_1))^s$ and $V_2 = ({X}(p_2))^s$, where, abusing notation a little bit, we consider here $(X(z))^s$ as a globally defined rational function on the auxiliary curve. 

As we discussed in the proof of Proposition~\ref{prop:irreducible}, the monodromy of the $\binom{r+s}{s}$-degree covering $\tilde{X}\colon C\to \mathbb{C}\mathrm{P}^1$ is induced by the action of the monodromy of $\tilde{X}=X^s$ considered as a $(r+s)$-degree covering $ \mathbb{C}\mathrm{P}^1\to  \mathbb{C}\mathrm{P}^1$,  on the functions $c_S$ indexed by the $s$-element subsets of points in the fiber. This immediately implies that all ramification points of  $\tilde{X}\colon C\to \mathbb{C}\mathrm{P}^1$ lie over the same four points $0,\infty,V_1,V_2$. Simple combinatorics of the action of a transposition (for the fibers over $V_1$ and $V_2$) or a product of disjoint cycles of length $r$ and $s$ (for the fibers over $0$ and $\infty$) on $s$-element subsets of an $(r+s)$-element set implies the following proposition:

\begin{proposition}\label{prop:ramificationprofile}
	The ramification profiles of $\tilde{X}\colon C\to \mathbb{C}\mathrm{P}^1$ over $V_1$ and $V_2$ are the same partitions of $\binom{r+s}{s}$ given by 
	\begin{align}
		(\, \underbrace{2,2,\ldots,2,2}_{\scriptscriptstyle\binom{r+s-2}{s-1}},\!\! \underbrace{1,1,\ldots,1,1}_{\scriptscriptstyle \binom{r+s-2}{s}+\binom{r+s-2}{s-2}} \!\!).
	\end{align}
	The ramification profiles of $\tilde{X}\colon C\to \mathbb{C}\mathrm{P}^1$ over $0$ and $\infty$ are given by the same partitions of $\binom{r+s}{s}$ with the parts $\ell_r\ell_s$, for $\ell_r| r$, $\ell_s| s$, with some multiplicities.
\end{proposition}

We omit full discussion of the partitions over $0$ and $\infty$, in particular, we give no combinatorial description of the multiplicities of their parts, since it will be irrelevant for our main application. 

\subsection{Quantum curve}\label{sec:quantum-curve}
In this subsection, we derive the corresponding quantum curve statement for the Ooguri-Vafa partition function that quantizes the true physical curve. The derivation is based on the quantum curve statement for the extended Ooguri-Vafa partition function~\cite[Section 10]{DPSS-HOMFLY} obtained by the quantization of the auxiliary curve and an adaptation of the construction of~\cite[Secton~7]{EGFMO2024} to shift operators.

We define the wave function as 
\begin{align}\label{eq:quantum-wave-functions}
\psi(\tilde X)&=Z|_{p_m=\tilde X^m}
=\sum_{\ell\geq0}Q^{r\ell/2}
H_{(\ell)}^{r,s}(Q^{-1/2},e^{\hbar/2})\tilde X^\ell
\end{align}
note that just one row partitions $\lambda=(l)$ contribute to the wave function, by the specialization $p_k=\tilde{X}^k$ the other contributions vanish.
Equivalently, 
\begin{align}\label{eq:quantum-integrated-correlators}
\log\psi(\tilde X)
=\sum_{\substack{g\geq0\\n\geq1}}
\frac{\hbar^{2g-2+n}}{n!}
\int_0^{\tilde X}\!\cdots\!\int_0^{\tilde X}
\left(\omega_n^{(g)}-
\delta_{(g,n),(0,2)}
\frac{d\tilde X_1d\tilde X_2}{(\tilde X_1-\tilde X_2)^2}\right).
\end{align}

Let $\widehat{\tilde{Y}}$ be the formal shift operator $\exp(\hbar \tilde{X} \partial_{\tilde X})$. Let $\widehat{\tilde{X}}$ denote the operator of multiplication by $\tilde{X}$. 

\begin{theorem}\label{thm:quantum-true-curve}
	There is an effectively computable
	finite shift operator
	\begin{align}\label{eq:quantum-true}
		\widehat F=
		\Big(\widehat{\tilde{X}}\Big)^{\frac 1{(r+s)} \binom{r+s}{s}}
		\left(\Big(\widehat{\tilde{Y}}\Big)^{\binom{r+s}{s}}+\sum_{j=0}^{\binom{r+s}{s}-1}a_j\Big(\widehat{\tilde X},Q^{1/2},e^{\hbar/(2s)}\Big)\Big(\widehat{\tilde{Y}}\Big)^j\right)
	\end{align}
	such that $\widehat F\psi=0$. The coefficients $a_j$ are rational
	functions regular at $\hbar=0$
	over $\mathbb C(Q^{1/2})(\tilde X)$. The semi-classical limit of $\widehat{F}$ is 
	\begin{align}\label{eq:quantum-symbol}
		\left.\widehat F\right|_{\hbar=0,\ \widehat{\tilde{X}}\to \tilde{X},\  \widehat{\tilde{Y}}\to \tilde{Y}}
		=F(\tilde X,\tilde Y).
	\end{align}
\end{theorem}

\begin{proof} We give an algorithm to construct $\widehat{F}$ rather than its explicit description. First, we recall the statement of~\cite[Section~10]{DPSS-HOMFLY}, in our notation. 
Let 
\begin{align}
	\psi^{\mathsf{ext}}(X)&=Z^{\mathsf{ext}}|_{p_k=X^k}
\end{align}
be the extended wave function. Then, with $\widehat{\tilde{Y}} = \exp(\frac\hbar s X\partial_X) $ and $\widehat{X}$ being multiplication by $X$, 
\begin{align}\label{eq:quantum-extended}
\left[\Big(\widehat{\tilde{Y}}\Big)^s-1-e^{\hbar/2}\widehat{X} \Big(\widehat{\tilde{Y}}\Big)^r
\left(Q^{-1/2}\Big(\widehat{\tilde{Y}}\Big)^s-Q^{1/2}\right)\right]
\psi^{\mathsf{ext}}(X)=0.
\end{align}

Using the determinant definition of the Schur functions, one can represent the wave function $\psi(\tilde X)$ in terms of $\psi^{\mathsf{ext}}(X)$, with $\tilde X = X^s$, as 
\begin{align}\label{eq:quantum-determinant}
\psi(\tilde X)=
\frac{\det\left((\xi^iX)^j
\psi^{\mathsf{ext}}\left(e^{\frac rs\hbar(j-s+1)}\xi^iX\right)
\right)_{i,j=0}^{s-1}}
{\det\left((\xi^iX)^j\right)_{i,j=0}^{s-1}}.
\end{align}

Following the idea of~\cite[Section~7]{EGFMO2024}, for a solution $f=f(X)$ of Equation~\eqref{eq:quantum-extended} we rewrite this equation as 
\begin{align}
 \Big(\widehat{\tilde{Y}}\Big)^{r+s}f
=
-Q^{1/2}e^{-\hbar/2}X^{-1}f
+Q\Big(\widehat{\tilde{Y}}\Big)^rf
+Q^{1/2}e^{-\hbar/2}X^{-1} \Big(\widehat{\tilde{Y}}\Big)^sf
\end{align}
and then encode this in the matrix form, acting on the vector $v\coloneqq (f, \widehat{\tilde{Y}}f,\cdots,\big(\widehat{\tilde{Y}}\big)^{r+s-1}f)^t$ as 
$\widehat{\tilde{Y}} v^i = \sum_{j=0}^{r+s-1}C^i_j(X) v^j$, where $0\leq i<r+s$. Here
\begin{align}
C^i_{j}(X)=
\begin{cases}
	\delta_{j,i+1},&0\leq i<r+s-1,\\[2mm]
	Q\delta_{j,r}
	+Q^{1/2}e^{-\hbar/2}X^{-1}
	(\delta_{j,s}-\delta_{j,0}),&i=r+s-1.
\end{cases}
\end{align}
Choose $p$ as in \eqref{eq:cS-definition}, so $rp\equiv1\pmod s$. The column vectors $v_i = (v_i^0,\dots,v_i^{r+s-1})^t$, $i=0,\dots,s-1$, with
$v_i^k = \xi^{pik}\big(\widehat{\tilde{Y}}\big)^k
	\psi^{\mathsf{ext}}(\xi^iX)$, $0\leq k<r+s$,
satisfy $\widehat{\tilde{Y}} v_i=\xi^{-pi} C(X)v_i$. To implement this in one equation, one pass from the $(r+s)$-dimensional space to the $s$-exterior power providing an $\binom{r+s}{s}$-dimensional linear difference system. 
Therefore,
\begin{align}\label{eq:quantum-exterior-system}
	\widehat{\tilde{Y}} \left( v_0\wedge\cdots\wedge v_{s-1} \right)=(-1)^{s-1}\big(\wedge^s C(X) \big) \left( v_0\wedge\cdots\wedge v_{s-1} \right)
\end{align}
(the sign here is implied by $\prod_{i=0}^{s-1}\xi^{pi}=(-1)^{s-1}$).

We now have to express $\psi(\tilde X)$ through the coordinates of the vector
$v_0\wedge\cdots\wedge v_{s-1}$ and then eliminate these coordinates.
To this end, extend the notation $v_i^k$ to all $k\in\mathbb Z$
by the same formula. Equation~\eqref{eq:quantum-extended} implies
\begin{align}\label{eq:quantum-component-recurrence}
v_i^{k+r+s}=Qv_i^{k+r}
+Q^{1/2}e^{-\hbar/2-k\hbar/s}X^{-1}
\left(v_i^{k+s}-v_i^k\right).
\end{align}
In particular, solving this equation for $v_i^k$ expresses every
component with $k<0$ as a rational linear combination of
$v_i^0,\ldots,v_i^{r+s-1}$. The coefficients of this linear
combination do not depend on $i$.

Substitute
$\psi^{\mathsf{ext}}(e^{\frac rs\hbar(j-s+1)}\xi^iX)
=\xi^{-pi r(j-s+1)}v_i^{r(j-s+1)}$
in Equation~\eqref{eq:quantum-determinant}. Since $pr\equiv1\pmod s$,
the powers of $X$ in the denominator cancel and we obtain
\begin{align}\label{eq:quantum-wave-components}
\psi(\tilde X)=
\frac{(-1)^{s-1}}{\det\left(\xi^{ij}\right)_{i,j=0}^{s-1}}
\det\left(v_i^{r(j-s+1)}\right)_{i,j=0}^{s-1}.
\end{align}
Using Equation~\eqref{eq:quantum-component-recurrence}, the determinant
on the right becomes a linear combination of the determinants
\begin{align}
\det \Big( v_i^{k_j} \Big)_{\substack{0\leq i<s\\1\leq j\leq s}},
\qquad 0\leq k_1<\cdots<k_s<r+s,
\end{align}
which are indeed the coordinates of
$v_0\wedge\cdots\wedge v_{s-1}$. Denote the resulting row of
coefficients, including the constant prefactor in
\eqref{eq:quantum-wave-components}, by $b_0(X)$. Thus
\begin{align}\label{eq:quantum-wave-row}
\psi(\tilde X)=b_0(X)
\left(v_0\wedge\cdots\wedge v_{s-1}\right).
\end{align}
This gives an explicit finite procedure to compute $b_0(X)$, which is a row with $\binom{r+s}{s}$ entries. Then, define rows $b_{j+1}(X)\coloneqq (-1)^{s-1}b_j(e^{\hbar/s}X)\big(\wedge^sC(X)\big)$, $j\geq 0$. Applying $\widehat{\tilde Y}$ and using
Equation~\eqref{eq:quantum-exterior-system}, we obtain
\begin{align}\label{eq:quantum-row-iteration}
\Big(\widehat{\tilde Y}\Big)^j\psi(\tilde X)
&=b_j(X)\left(v_0\wedge\cdots\wedge v_{s-1}\right).
\end{align}

Now, in full analogy with~\cite[Section~7.1]{EGFMO2024}, we find the coefficients
$a_j(X)$ by solving the linear system
\begin{align}\label{eq:quantum-elimination}
b_{\binom{r+s}{s}}(X)+\sum_{j=0}^{\binom{r+s}{s}-1}a_j(X)b_j(X)=0.
\end{align}
Note that a priori $a_j$ can depend on $X$ itself, not exclusively on $\tilde X = X^s$. However, a solution of this
system gives the required equation $\widehat F\psi=0$. It still remains
to show that this system has a unique solution, that the coefficients $a_j$, $j=0,\dots, \binom{r+s}{s}-1$, are regular at $\hbar=0$ and depend only on $X^s$, and that the semi-classical limit of $\widehat{F}$ is $F$.

Let us first put $\hbar=0$. The eigenvalues $t_1,\ldots,t_{r+s}$
of $C(X)$ satisfy
\begin{align}\label{eq:quantum-classical-eigenvalues}
t^s-1-Xt^r\left(Q^{-1/2}t^s-Q^{1/2}\right)=0,
\end{align}
and the corresponding eigenvectors are
$(1,t_i,\ldots,t_i^{r+s-1})^t$. Putting $z_i=Xt_i^r$ gives
\begin{align*}
t_i^s=\frac{1-Q^{1/2}z_i}{1-Q^{-1/2}z_i}.
\end{align*}
Hence the $z_i$ are the roots of $P$ in
Equation~\eqref{eq:z-root-polynomial}. Since $r$ and $s$ are coprime, this correspondence
between the $t_i$'s and the $z_i$'s is bijective. By
Equations~\eqref{eq:cS-powers},
\begin{align}\label{eq:quantum-exterior-eigenvalues}
(-1)^{s-1}\prod_{i\in S}t_i=c_S,\qquad |S|=s.
\end{align}
These are the eigenvalues of $(-1)^{s-1}\wedge^sC(X)$, and they
are generically distinct by Proposition~\ref{prop:irreducible}.

The value of $b_0(X)$ on the exterior product of eigenvectors
indexed by $S\subset\{1,\dots,r+s\}$, $|S|=s$, is a nonzero constant times
\begin{align}
\det\Big(t_i^{r(j-s+1)}\Big)_{\substack{i\in S\\0\leq j<s}}.
\end{align}
This determinant is proportional (with a nonvanishing factor) to the Vandermonde determinant in the distinct values
$t_i^r=z_i/X$. At $\hbar=0$, the recursion that defines $b_{j}$ in terms of $b_{j-1}$, $j\geq 1$,
implies that the value of $b_j$ on this eigenvector
is $c_S^j$ times the nonzero value of $b_0$ on the same eigenvector. The Vandermonde
determinant in the $c_S$ now shows that the first
$\binom{r+s}{s}$ rows $b_j$ are linearly independent.
Hence, the system~\eqref{eq:quantum-elimination} has a unique
solution, which is rational in $X,Q^{1/2},e^{\hbar/(2s)}$ and regular at
$\hbar=0$. Moreover, in the classical limit $\hbar\to 0$ it gives
\begin{align}
\tilde Y^{\binom{r+s}{s}}
+\sum_{j=0}^{\binom{r+s}{s}-1}
\left.a_j\right|_{\hbar=0}\tilde Y^j
=\prod_{|S|=s}(\tilde Y-c_S).
\end{align}
Multiplication by
$\tilde X^{\frac1{r+s}\binom{r+s}{s}}$ gives precisely $F$, hence the semi-classical limit of $\widehat{F}$ is indeed $F$. This is the analogue of the classical-limit
argument in~\cite[Theorem~7.2]{EGFMO2024}.

Finally, let us prove that $a_j$ depend on $X$ only through
$\tilde X=X^s$. To this end, notice that 
\begin{align}
C(\xi X)=\xi^{-p}D^{-1}C(X)D,
\end{align}
where $D=\operatorname{diag}(1,\xi^p,\ldots,\xi^{p(r+s-1)})$.
Taking the $s$-th exterior power cancels the scalar
$\xi^{-ps}=1$. Let us show that the recurrence
\eqref{eq:quantum-component-recurrence} and the determinant in
\eqref{eq:quantum-wave-components} similarly give
\begin{align}
b_0(\xi X)=(-1)^{s-1}b_0(X)\big(\wedge^sD\big).
\end{align}
Indeed, replacing $X$ by $\xi X$ in the recurrence is compatible
with multiplying its component of index $k$ by $\xi^{-pk}$.
For the indices $r(j-s+1)$ the product of these factors is
$\xi^{prs(s-1)/2}=(-1)^{s-1}$.

The recursion defining $b_j$'s then gives the same
transformation rule for each of them, $j\geq 0$.  Uniqueness of the solution
of \eqref{eq:quantum-elimination} implies
$a_j(\xi X)=a_j(X)$. Thus $a_j$ depend on $X$ through $\tilde X=X^s$, as required.
\end{proof}

\begin{remark} It would be interesting to compare this statement on quantum curve with the one recently obtained in~\cite{CEL2026}. A priori there is no reason to expect that the operator in \emph{op.~cit.} matches exactly the one we construct here. 
\end{remark}

\section{Augmentation ideal} \label{sec:Augmentation}

Fix $r,s\geq 1$ and assume they are coprime. 
Here we present a construction of the augmentation variety, following~\cite{Ng-Introduction}, which is already adapted here specifically for the case of torus knots given by the closure of the braid $B^s$ on $r$ strings, where $B=\sigma_1\cdots\sigma_{r-1}$ (some of these specializations are left as exercises for the reader in~\emph{op.~cit.}).

Take the algebra $\mathcal A=\mathbb C(Q^{1/2})[a_{ij},\tilde X^{\pm1},\tilde Y^{\pm1}]$, where $1\le i\ne j\le r$.
Thus $\tilde X\tilde Y\ne0$ throughout the augmentation construction. Associate to it
\begin{align}\label{PhiBLdefinetion}
	\Phi_B^L \coloneqq 
	\begin{pmatrix}
		-a_{21}&1&0&\cdots&0\\
		-a_{31}&0&1&\cdots&0\\
		\vdots&\vdots&\ddots&\ddots&\vdots\\
		-a_{r1}&0&\cdots&0&1\\
		1&0&\cdots&0&0
	\end{pmatrix},
	\qquad
	\Phi_B^R \coloneqq
	\begin{pmatrix}
		-a_{12}&-a_{13}&\cdots&-a_{1r}&1\\
		1&0&\cdots&0&0\\
		0&1&\cdots&0&0\\
		\vdots&\ddots&\ddots&\vdots&\vdots\\
		0&\cdots&0&1&0
	\end{pmatrix}.
\end{align}
We assume that for $r=1$, \(\Phi_B^L=\Phi_B^R=(1)\).
Let 
\begin{align}
	\Phi^L_{B^s}
	=
	\phi_B^{s-1}(\Phi_B^L)\cdots \phi_B(\Phi_B^L)\Phi_B^L;
	\qquad 
	\Phi^R_{B^s} 
	=
	\Phi_B^R \phi_B(\Phi_B^R)\cdots\phi_B^{s-1}(\Phi_B^R),
\end{align}
where the action of $\phi_B$ on $a_{ij}$ is given by
\begin{align} \label{eq:action}
	\phi_{B}(a_{ij})
=
\begin{cases}
	a_{i+1,j+1},&i,j<r,\\
	-a_{i+1,1},&j=r,\\
	-a_{1,j+1},&i=r,
\end{cases}
\end{align}
and 
\begin{align}
	\phi^m_B(a_{ij})
	=
	(-1)^{
		\left\lfloor\frac{i+m-1}{r}\right\rfloor+
		\left\lfloor\frac{j+m-1}{r}\right\rfloor}
	a_{\overline{i+m},\,\overline{j+m}},
\end{align}
where \(\overline{k}\) is the representative of \(k\pmod r\) in \(\{1,\ldots,r\}\).	Here, $\phi_B$ is called the \emph{braid action} \(\phi_B:\mathcal A\to\mathcal A\) which defines an algebra automorphism over \(R=\mathbb C(Q^{1/2})[\widetilde X^{\pm1},\widetilde Y^{\pm1}]\). In particular, for all \(a,b\in\mathcal A\) and \(c\in R\), it satisfies
\[
\begin{aligned}
\phi_B(a+b)&=\phi_B(a)+\phi_B(b),\\
\phi_B(ca)&=c\,\phi_B(a),\\
\phi_B(ab)&=\phi_B(a)\phi_B(b),\\
\phi_B(1)&=1.
\end{aligned}
\]
Moreover, \(\phi_B\) is bijective, with inverse \(\phi_B^{-1}=\phi_B^{r-1}\), since \(\phi_B^r=\mathrm{id}\). The automorphism $\phi_B$ extends entrywise to matrices.

Let 
$
\Lambda \coloneqq
	\operatorname{diag}\!\left(
	Q^{\,r-\frac s2}\tilde X\,\tilde Y^s,\,
	1,\ldots,1
	\right)
$ be an $r\times r$ matrix, and define
\begin{equation}\label{eq:aug-A-matrices}
	A_{ij}=\begin{cases}
		\tilde{Y} a_{ij},&i<j,\\
		-Qa_{ij},&i>j,\\
		\tilde{Y}-Q,&i=j,
	\end{cases}
	\qquad
	\widehat A_{ij}=\begin{cases}
		\tilde{Y} Q a_{ij},&i<j,\\
		-Q a_{ij},&i>j,\\
		Q\tilde{Y} -Q,&i=j.
	\end{cases}
\end{equation}

Consider the ideal in $\mathcal{A}$ generated by the entries of 
\begin{equation}\label{eq:aug-two-matrix-equations1}
	\widehat A-\Lambda\Phi_{B^s}^L A,
	\qquad \text{and} \qquad 
	A\Lambda-\widehat A\Phi_{B^s}^R.
\end{equation}
We refer to the corresponding matrix equations as the \emph{augmentation equations}
\begin{equation}\label{eq:aug-two-matrix-equations}
	\widehat A-\Lambda\Phi_{B^s}^L A=0,
	\qquad \text{and} \qquad 
	A\Lambda-\widehat A\Phi_{B^s}^R=0.
\end{equation}
This ideal, denoted by $I$, defines the augmentation variety with $\tilde X\tilde Y\ne0$; points are taken over an algebraic closure of $\mathbb C(Q^{1/2})$.
The augmentation curve is the Zariski closure of its projection in the torus $(\tilde X,\tilde Y)$.
We use the closure of this curve in the affine plane to define its polynomial equation.
\begin{proposition} \label{prop:augmentation-projection}
The closure of the augmentation curve in the affine plane is an irreducible curve given by an irreducible polynomial $\tilde F(\tilde X,\tilde Y)=0$.
\end{proposition} 
Algebraically, consider the ideal $I\cap \mathbb C(Q^{1/2})[\tilde X,\tilde Y]$.
The radical of $I\cap\mathbb C(Q^{1/2})[\tilde X,\tilde Y]$ is the principal ideal $(\tilde F)$. 

\begin{remark}To compute this contraction in a polynomial ring, adjoin one variable $h$ and the relation $h\tilde X\tilde Y-1$ to the entries in the augmentation equations~\eqref{eq:aug-two-matrix-equations}.
Compute a Gr{\"o}bner basis for the lexicographic order
	\begin{align}
		h>a_{12}>a_{13}>\cdots>a_{r,r-1}>\tilde X>\tilde Y.
	\end{align}
	Eliminate $h$ and all $a_{ij}$ by retaining only the basis elements involving $\tilde X,\tilde Y$, and take the square-free part of their greatest common divisor. Then normalize the resulting polynomial \(\tilde F\) so that its leading coefficient in this lexicographic order is \(1\).
\end{remark}

Proposition~\ref{prop:augmentation-projection} is proved below, along the proof of the following theorem:

\begin{theorem} \label{thm:FF} 
The polynomial $\tilde F$ coincides with $F$ given by Eq.~\eqref{eq:F}.
\end{theorem}

\begin{remark} This theorem proves a conjecture stated in~\cite{GJKS}, see Eq.~(2.36) in \emph{op.~cit.}
\end{remark}

\begin{remark} The full combinatorial analysis of the augmentation curve for $Q=1$ is available in~\cite{Cornwell}. 
\end{remark}

\begin{remark} Some examples of the reconstruction procedure of augmentation polynomial from the disk-series data of torus knots, using the back then conjectural degree bounds of~\cite{JKS}, are available in~\cite{Mahowald2017}.
\end{remark}

\subsection{Proof of Theorem~\ref{thm:FF}} 
We have to show the following two properties:
\begin{itemize}
    \item[(I)] \hypertarget{augmentation-to-F}{} every augmentation point in the torus satisfies $F(\tilde X,\tilde Y)=0$
    \item[(II)]\hypertarget{F-to-augmentation}{} every generic point of $F=0$ can be lifted to an augmentation by reconstructing suitable $a_{ij}$.
\end{itemize}
	
For the first implication \hyperlink{augmentation-to-F}{(I)}, let's start with a tuple $(\tilde X,\tilde Y,a_{ij})$ satisfying the augmentation equations~\eqref{eq:aug-two-matrix-equations}, with $\tilde X\tilde Y\ne0$.
Let $U$ and $D$ be upper and lower triangular matrices, respectively, with units on the diagonal and with entries $a_{ij}$ above or below the diagonal. In particular, $A=\tilde{Y}U-QD$ and $\widehat{A}=\tilde{Y}QU-QD$. Let $Z\coloneqq \tilde{Y}^{-1} Q^{1/2} DU^{-1}$.
	
\begin{lemma} \label{lem:Pol-Z} The augmentation equations~\eqref{eq:aug-two-matrix-equations} imply that $P(Z)=0$. Moreover, the characteristic polynomial of $Z$ divides $P$. 
\end{lemma}

Let $z_1,\ldots,z_{r+s}$ be the roots of $P(z)$ in the algebraic closure,
counted with multiplicities. Lemma~\ref{lem:Pol-Z} implies that $r$ of
them, counted with their algebraic multiplicities, are the eigenvalues
of $Z$, the matrix determined by the given tuple $(\tilde{X},\tilde{Y},a_{ij})$. 
Let $S_0\subset\{1,\dots,r+s\}$ denote the subset that indexes the remaining $s$ roots, that is, the roots of $P(z)/\det(z\mathop{\mathrm{Id}} - Z)$. Though we don't have direct access to these $s$ remaining roots, we can derive their properties from the properties of $z_i$, $i\in S_0^c\coloneqq \{1,\dots,r+s\} \setminus S_0$:

\begin{lemma} \label{lem:properties-Z-eigenvalues} The augmentation equations~\eqref{eq:aug-two-matrix-equations} imply that the eigenvalues of $Z$ satisfy
	\begin{align}
		\prod_{i\in S_0^c} z_i & = Q^{r/2} \tilde{Y}^{-r}; &
		\prod_{i\in S_0^c} \frac{1-Q^{-1/2}z_i}{1-Q^{1/2}z_i} & = (-1)^{r-1}Q^{-s/2}\tilde{X}\tilde{Y}^s.
	\end{align}
\end{lemma}

This is sufficient to derive that $c_{S_0} = \tilde{Y}$. Indeed, from Equations~\eqref{eq:total-products} and Lemma~\ref{lem:properties-Z-eigenvalues} we have 
\begin{align}
(-1)^{s-1}\frac{\prod_{i\in S_0}z_i}{\tilde{X}}& =\tilde{Y}^r,
& \prod_{i\in S_0}\frac{1-Q^{-1/2}z_i}{1-Q^{1/2}z_i}& =\tilde{Y}^{-s}.
\end{align}
Then Equation~\eqref{eq:cS-definition} (or Equation~\eqref{eq:cS-powers}) gives $c_{S_0}=\tilde{Y}^{rp-sq}=\tilde{Y}$. This means that we indeed necessarily have $F(\tilde{X},\tilde{Y}) = 0$ for any $(\tilde X, \tilde Y, a_{ij})$ satisfying the augmentation	equations~\eqref{eq:aug-two-matrix-equations} with $\tilde{X}\tilde{Y}\ne 0$.  This establishes \hyperlink{augmentation-to-F}{(I)}.

Next we have to show \hyperlink{F-to-augmentation}{(II)}, which is finally proved in Sec.~\ref{sec:final}. To reconstruct the $a_{ij}$ from a generic pair $(\tilde X,\tilde Y)$ satisfying $F(\tilde X,\tilde Y)=0$, we will give an explicit formula.
First, we rewrite the augmentation equations~\eqref{eq:aug-two-matrix-equations} and prove Lemmata~\ref{lem:Pol-Z} and~\ref{lem:properties-Z-eigenvalues}.

\subsubsection{Diagonal change of basis} \label{sec:diagonalbasis}

We introduce a diagonal change of basis in $U,D,Z$ that gives better control on their entries. First, let us list a number of immediate corollaries of the braid action~\eqref{eq:action} and the augmentation equations~\eqref{eq:aug-two-matrix-equations}. Direct multiplication gives
\begin{align}\label{eq:universal-triangular-identities}
	\phi_B(U)=\Phi_B^L U\Phi_B^R,\qquad
	\phi_B(D)=\Phi_B^L D\Phi_B^R.
\end{align}
Iterating these identities gives
\begin{align}\label{eq:universal-A-identity}
	\phi_B^s(A)=\Phi_{B^s}^L A\Phi_{B^s}^R.
\end{align}
In particular, the two augmentation equations imply
$A\Lambda=\Lambda\phi_B^s(A)$. Comparing the strictly
upper and strictly lower entries gives 
	\begin{align}\label{eq:covariance}
		U\Lambda = \Lambda \phi_B^s (U); \qquad  D\Lambda  = \Lambda \phi_B^s (D). 
	\end{align}
	
Introduce auxiliary matrices
	\begin{align}
	J& \coloneqq\begin{pmatrix}0&1&&0\\&0&\ddots&\\&&\ddots&1\\-1&0&\cdots&0\end{pmatrix},
	& N& \coloneqq \begin{pmatrix}0&1&&0\\&0&\ddots&\\&&\ddots&1\\0&0&\cdots&0\end{pmatrix},
	& N^{\mathsf{T}}& \coloneqq \begin{pmatrix}0&&&0\\1 &0& &\vdots \\&\ddots&\ddots&0\\0& & 1 &0\end{pmatrix}.
	\end{align}
	For $r=1$ these mean $J=(-1)$ and $N=N^{\mathsf{T}}=(0)$. In terms of matrix $J$, the formula for the braid action \eqref{eq:action} is equivalent to
	\begin{align}
	\phi_B(U+D)=J(U+D)J^{-1}.
	\end{align}
	Hence \eqref{eq:covariance} is equivalent to the requirement that $U+D$ commutes with $\Lambda J^s$, that is $[U+D,\Lambda J^s]=0$.
	
The diagonal change of basis is given by a matrix $E$ that we are going to construct now. 
Choose a nonzero $\kappa$ in the algebraic closure satisfying the equation
\begin{equation}\label{eq:kappa}
	\kappa^r=Q^{r-s/2}\tilde{X}\tilde{Y}^s=\det\Lambda.
\end{equation}
There is an invertible diagonal matrix $E$, with $E_{11}=1$, such that
\begin{equation}\label{eq:E}
	E^{-1}\Lambda J^sE=\kappa J^s
\end{equation}
replacing the matrix $\Lambda$ by the scalar $\kappa$ through diagonal change of basis.
For an explicit construction, let $\overline{i}$ denote the representative
modulo $r$ in $\{1,\ldots,r\}$ and impose successively
$E_{\overline{i+s},\overline{i+s}}
	=\kappa\Lambda_{ii}^{-1}E_{ii}$. 
Since $r$ and $s$ are coprime, this determines $E$.
After one complete cycle the multiplier is $\kappa^r/\det\Lambda=1$, so the recursion is consistent.
	
Let $U'\coloneqq E^{-1} UE$ and $D'\coloneqq E^{-1} DE$. Notice that now $U'+D'$ commutes with $J^s = \kappa^{-1}  E^{-1}\Lambda J^sE$. Since $r$ and $s$ are coprime and $J^r = -\mathop{\mathrm{Id}}$, $U'+D'$ commutes with $J$. 
Every matrix commuting with $J$ is a polynomial in $J$ of degree at most $r-1$, hence there exist scalars $b_1,\ldots,b_{r-1}$ such that
\begin{align}
	 U'+D'& =2\mathop{\mathrm{Id}}+\sum_{j=1}^{r-1}b_jJ^j, 
	 & U'=\mathop{\mathrm{Id}}+\sum_{j=1}^{r-1}b_jN^j,\qquad
	 D'=\mathop{\mathrm{Id}}-\sum_{j=1}^{r-1}b_{r-j}(N^{\mathsf T})^j.
\end{align}

The advantage of these new $U'$ and $D'$, and the corresponding 
\begin{align} \label{eq:aij-accent}
a_{ij}' \coloneqq E_{ii}^{-1} a_{ij} E_{jj} = \begin{cases}
	b_{j-i},&i<j,\\
	-b_{r-i+j},&i>j.
\end{cases}
\end{align}
is that substituting these entries in the right-hand side of Equation~\eqref{eq:action} returns the same entries; this means that $a_{ij}'$ are fixed by the braid action $\phi_B$. Moreover, we now have $r-1$ augmentation parameters $b_1,\ldots,b_{r-1}$ rather than the original $r(r-1)$ parameters $a_{ij}$.
The main goal now is to show that the augmentation equations~\eqref{eq:aug-two-matrix-equations} can be equivalently rewritten in terms of $b_1,\dots,b_{r-1}$ (and thus $a_{ij}$ can be reconstructed from them via Equation~\eqref{eq:aij-accent}). 

\subsubsection{Equivalent form of augmentation equations}  \label{sec:equivalenceaug}

To avoid any confusion, throughout this subsection we use the notation $\Phi_B^L(a)$ for the matrix defined in \eqref{PhiBLdefinetion}, whereas $\Phi_B^L(a')$ denotes the same matrix with each $a_{ij}$ replaced by $a_{ij}'$ as defined in \eqref{eq:aij-accent}. For brevity, let $T$ be $\Phi_B^L(a')$, explicitly,
\begin{align}\label{eq:T-companion}
T&=\begin{pmatrix}
b_{r-1}&1&0&\cdots&0\\
b_{r-2}&0&1&\cdots&0\\
\vdots&\vdots&\ddots&\ddots&\vdots\\
b_1&0&\cdots&0&1\\
1&0&\cdots&0&0
\end{pmatrix},
&\det(t\mathop{\mathrm{Id}}-T)
&=t^r-\sum_{j=1}^{r-1}b_jt^j-1.
\end{align}
For $r=1$ these formulas mean $T=(1)$ and
$\det(t\mathop{\mathrm{Id}}-T)=t-1$. Notice that
$\det T=(-1)^{r-1}$, so $T$ is invertible.

The statement at the end of the previous subsection that the braid action fixes $a'_{ij}$ means that
$\left.\phi_B(a_{ij})\right|_{a_{kl}\to a'_{kl}}=a'_{ij}$. Thus, if $\Phi_{B^m}^L(a')$ denotes the matrix
obtained by substituting $a'_{ij}$ for $a_{ij}$ in $\Phi_{B^m}^L$, then
\begin{align}\label{eq:primed-braid-power}
\Phi_{B^m}^L(a')=T^m, \qquad m\ge1.
\end{align}

Let us prove that we also have
\begin{align}\label{eq:T-upper-lower}
T^rD'=U',\qquad T^r=U'D'^{-1}.
\end{align}
To check this directly, let $e_1,\ldots,e_r$ be the standard basis vectors.
The shape of $T$ gives $Te_j=e_{j-1}$ for $j\ge2$, hence
$e_j=T^{r-j}e_r$. Moreover,
$T^r=\mathop{\mathrm{Id}}+\sum_{k=1}^{r-1}b_kT^k$.
Using the entries of $D'$, we obtain, for $1\le j\le r$,
\begin{align}
D'e_j
&=\left(T^{r-j}-\sum_{k=j}^{r-1}b_kT^{k-j}\right)e_r,
& T^jD'e_j
&=\left(\mathop{\mathrm{Id}}+\sum_{k=1}^{j-1}b_kT^k\right)e_r.
\end{align}
Multiplication of the second equality by $T^{r-j}$ gives
$T^rD'e_j=U'e_j$, which proves Equation~\eqref{eq:T-upper-lower}.

We now express the original braid matrix $\Phi_B^L(a)$ in terms of $T$.
Direct substitution in $\Phi_B^L(a)$, followed by induction in the formula
for $\Phi_{B^m}^L(a)$, gives
\begin{align}\label{eq:diagonal-braid-matrix}
\Phi_{B^m}^L(a')
=(J^mEJ^{-m})^{-1}\Phi_{B^m}^L(a) E.
\end{align}
Indeed, $J^mEJ^{-m}$ is diagonal, with $i$-th diagonal entry
$E_{\overline{i+m},\overline{i+m}}$; these are exactly the diagonal
factors at the shifted indices in Equation~\eqref{eq:action}.
The factors at consecutive indices cancel in the product defining
$\Phi_{B^m}^L(a)$.
Equation~\eqref{eq:E} implies
$E^{-1}\Lambda(J^sEJ^{-s})=\kappa\mathop{\mathrm{Id}}$.
Consequently,
\begin{align}\label{eq:Lambda-Phi-T}
E^{-1}\Lambda\Phi_{B^s}^L(a) E=\kappa T^s.
\end{align}

Conjugating the first augmentation equation \eqref{eq:aug-two-matrix-equations} by $E$, we therefore obtain
\begin{align}\label{eq:augmentation-primed}
Q\tilde Y U'-QD'=\kappa T^s(\tilde Y U'-QD').
\end{align}
Multiplication on the right by $D'^{-1}$ and
Equation~\eqref{eq:T-upper-lower} transform it into
\begin{align}\label{eq:augmentation-companion}
\ \kappa\tilde Y T^{r+s}-\kappa Q T^s
-Q\tilde Y T^r+Q\mathop{\mathrm{Id}}=0.
\end{align}
This is the desired equation for the remaining $r-1$ augmentation parameters $b_1,\ldots,b_{r-1}$:

\begin{lemma} 
Choose $\kappa$ and $E$ as in
Equations~\eqref{eq:kappa} and~\eqref{eq:E}, and express
$a_{ij}$ in terms of $b_1,\ldots,b_{r-1}$ by
Equation~\eqref{eq:aij-accent}. Then Equation~\eqref{eq:augmentation-companion} is equivalent
to both augmentation equations~\eqref{eq:aug-two-matrix-equations}.
\end{lemma}

\begin{proof}
We have already used the first augmentation equation \eqref{eq:aug-two-matrix-equations} to derive~\eqref{eq:augmentation-companion}. 
Let us check that Equation~\eqref{eq:augmentation-companion} gives both augmentation equations after reversing the change of variables. Start with $b_1,\ldots,b_{r-1}$ satisfying
Equation~\eqref{eq:augmentation-companion}, and construct $U',D',E$
as above. Their form implies that $U'+D'$ commutes with $J$, hence
that $U+D$ commutes with $\Lambda J^s$ after reversing the diagonal
change. Thus Equation~\eqref{eq:covariance} holds.
Reversing the steps from Equation~\eqref{eq:augmentation-primed}
to Equation~\eqref{eq:augmentation-companion} gives
$\widehat A=\Lambda\Phi_{B^s}^L(a) A$. Finally,
\begin{align}
A\Lambda=\Lambda\phi_B^s(A)
=\Lambda\Phi_{B^s}^L(a) A\Phi_{B^s}^R(a)
=\widehat A\Phi_{B^s}^R(a).
\end{align}
This proves the equivalence of~\eqref{eq:augmentation-companion} with both augmentation equations.
\end{proof}

\subsubsection{Diagonal change of basis applied to \texorpdfstring{$Z$}{Z}} \label{sec:diagonalchangeZ}

Let $Z'\coloneqq E^{-1}ZE$. Equation~\eqref{eq:T-upper-lower} gives
\begin{align}\label{eq:Z-prime-T}
Z'=\frac{Q^{1/2}}{\tilde Y}D'U'^{-1}
=\frac{Q^{1/2}}{\tilde Y}T^{-r}.
\end{align}
We will express $T^s$ in terms of $Z'$ as well.
First, subtracting $\tilde Y U'-QD'$ from the two sides of
Equation~\eqref{eq:augmentation-primed} gives
\begin{align}
(\kappa T^s-\mathop{\mathrm{Id}})(\tilde Y U'-QD')
=(Q-1)\tilde Y U'.
\end{align}
The right-hand side is invertible, hence $\tilde Y U'-QD'$ is invertible.
Since
\begin{align}
\tilde Y U'-QD'
=\tilde Y(\mathop{\mathrm{Id}}-Q^{1/2}Z')U',
\end{align}
$\mathop{\mathrm{Id}}-Q^{1/2}Z'$ is invertible too.
Thus Equation~\eqref{eq:augmentation-primed} can be rewritten as
\begin{align}\label{eq:T-s-Z-prime}
\kappa T^s
=Q(\mathop{\mathrm{Id}}-Q^{-1/2}Z')
(\mathop{\mathrm{Id}}-Q^{1/2}Z')^{-1}.
\end{align}
The left-hand side is invertible, so
$\mathop{\mathrm{Id}}-Q^{-1/2}Z'$ is also invertible.

\begin{proof} [Proof of Lemma~\ref{lem:Pol-Z}] Taking the $r$-th power in Equation~\eqref{eq:T-s-Z-prime} and using
Equations~\eqref{eq:kappa} and~\eqref{eq:Z-prime-T}, we get
\begin{align}
\tilde X Z'^{-s}
=(\mathop{\mathrm{Id}}-Q^{-1/2}Z')^r
(\mathop{\mathrm{Id}}-Q^{1/2}Z')^{-r}.
\end{align}
All factors here commute, since they are rational expressions in $Z'$.
Consequently,
\begin{align}\label{eq:P-Z-prime}
Z'^s(\mathop{\mathrm{Id}}-Q^{-1/2}Z')^r
-\tilde X(\mathop{\mathrm{Id}}-Q^{1/2}Z')^r=0.
\end{align}
This is precisely $P(Z')=0$. Conjugation by $E$ gives $P(Z)=0$.

To prove the divisibility in Lemma~\ref{lem:Pol-Z}, we need a little more.
Notice that both $T^r$ and $T^s$ belong to the algebra of matrices
generated by $Z'$, by Equations~\eqref{eq:Z-prime-T}
and~\eqref{eq:T-s-Z-prime}.
Here the inverse of an invertible matrix is a polynomial in that matrix,
by the Cayley--Hamilton theorem.
Using the same integers $p,q$ with $rp-sq=1$ as in Section~\ref{sec:TrueCurve}, we obtain
\begin{align}
T=(T^r)^p(T^s)^{-q}.
\end{align}
Thus $T$ also belongs to the algebra generated by $Z'$.
Conversely, Equation~\eqref{eq:Z-prime-T} shows that $Z'$ belongs to
the algebra generated by $T$. The two algebras therefore coincide.

The vectors $e_r,Te_r,\ldots,T^{r-1}e_r$ are
$e_r,e_{r-1},\ldots,e_1$, so they form a basis.
Since the two algebras coincide, the vectors obtained by applying
polynomials in $Z'$ to $e_r$ span the same space.
It is enough to use powers less than $r$.
Hence $e_r,Z'e_r,\ldots,Z'^{r-1}e_r$ form a basis too.
This means that $Z'$ is cyclic.
In particular, no polynomial of degree less than $r$ annihilates $Z'$,
so its minimal polynomial is its characteristic polynomial.
Since $P(Z')=0$, we conclude that
\begin{align}\label{eq:characteristic-divides-P}
\det(z\mathop{\mathrm{Id}}-Z)
=\det(z\mathop{\mathrm{Id}}-Z')\quad\text{divides}\quad P(z).
\end{align}
This proves Lemma~\ref{lem:Pol-Z} on the open set under consideration,
including the cases when some roots coincide.
\end{proof}

\begin{proof}[Proof of Lemma~\ref{lem:properties-Z-eigenvalues}]
We now prove Lemma~\ref{lem:properties-Z-eigenvalues}.
Recall that all roots are counted with their multiplicities.
Since $U'$ and $D'$ have determinant one,
\begin{align}
\prod_{i\in S_0^c}z_i=\det Z'
=Q^{r/2}\tilde Y^{-r}.
\end{align}
Taking determinants in Equation~\eqref{eq:T-s-Z-prime} gives
\begin{align}
\kappa^r(-1)^{s(r-1)}
=Q^r\prod_{i\in S_0^c}
\frac{1-Q^{-1/2}z_i}{1-Q^{1/2}z_i}.
\end{align}
By Equation~\eqref{eq:kappa}, this is equivalent to
\begin{align}
\prod_{i\in S_0^c}\frac{1-Q^{-1/2}z_i}{1-Q^{1/2}z_i}
=(-1)^{s(r-1)}Q^{-s/2}\tilde X\tilde Y^s
=(-1)^{r-1}Q^{-s/2}\tilde X\tilde Y^s.
\end{align}
The last equality follows because $(r-1)(s-1)$ is even for coprime $r,s$.
This proves the lemma.
\end{proof}

These arguments also justify the use of the remaining $s$ roots in the
preceding proof. They are the roots of the polynomial quotient in
Equation~\eqref{eq:characteristic-divides-P}; their multiplicities are
those left after removing the characteristic polynomial of $Z$.
For $\tilde X\ne0$, none of the roots of $P$ is $0,Q^{1/2},Q^{-1/2}$.
Moreover, the leading coefficient of $P$ is the nonzero constant
$(-Q^{-1/2})^r$, so its roots remain finite at every finite value of
$\tilde X$. Hence the symmetric expression defining $F$ specializes
regularly at every nonzero $\tilde X$, even when roots coincide.
The argument giving $c_{S_0}=\tilde Y$ and $F(\tilde X,\tilde Y)=0$
therefore applies to all augmentation points with $\tilde X\tilde Y\ne0$.
	
\subsubsection{Final steps of the proof of Theorem~\ref{thm:FF}}\label{sec:final}

It remains to prove \hyperlink{F-to-augmentation}{(II)} by reconstructing suitable \(a_{ij}\) for a generic pair \((\tilde X,\tilde Y)\) satisfying \(F(\tilde X,\tilde Y)=0\).
We can assume that $\tilde X\tilde Y\ne0$ and that $P$ has distinct roots.
By Equation~\eqref{eq:F}, there is a subset $S_0$ of cardinality $s$
for which $c_{S_0}=\tilde Y$. Dividing the total products in
Equation~\eqref{eq:total-products} by the products over $S_0$ given by
Equation~\eqref{eq:cS-powers}, we obtain
\begin{align}\label{eq:complementary-root-products}
\prod_{i\in S_0^c}z_i&=Q^{r/2}\tilde Y^{-r},
&\prod_{i\in S_0^c}\frac{1-Q^{-1/2}z_i}{1-Q^{1/2}z_i}
&=(-1)^{r-1}Q^{-s/2}\tilde X\tilde Y^s.
\end{align}

Choose $\kappa$ satisfying Equation~\eqref{eq:kappa}, and, for each
$i\in S_0^c$, define
\begin{align}\label{eq:companion-roots}
t_i=
\left(\frac{Q^{1/2}}{\tilde Y z_i}\right)^p
\left(\frac Q\kappa
\frac{1-Q^{-1/2}z_i}{1-Q^{1/2}z_i}\right)^{-q}.
\end{align}
Here again $rp-sq=1$.
The root equation $P(z_i)=0$ and Equation~\eqref{eq:kappa} give
\begin{align}
\left(\frac{Q^{1/2}}{\tilde Y z_i}\right)^s
=\frac{Q^{s/2}}{\tilde Y^s\tilde X}
\left(\frac{1-Q^{-1/2}z_i}{1-Q^{1/2}z_i}\right)^r
=\left(\frac Q\kappa
\frac{1-Q^{-1/2}z_i}{1-Q^{1/2}z_i}\right)^r.
\end{align}
Thus the same calculation as for Equation~\eqref{eq:cS-powers} shows that
\begin{align}\label{eq:companion-root-powers}
t_i^r=\frac{Q^{1/2}}{\tilde Y z_i},\qquad
t_i^s=\frac Q\kappa
\frac{1-Q^{-1/2}z_i}{1-Q^{1/2}z_i}.
\end{align}
These two equations determine $t_i$ uniquely: the ratio of two solutions
would have both its $r$-th and $s$-th powers equal to one.
Also, the $t_i$ are distinct, since equality of two of them would imply
equality of the corresponding $z_i$.

Taking products in Equation~\eqref{eq:companion-root-powers}, and using
Equation~\eqref{eq:complementary-root-products}, gives
\begin{align}
\Big(\prod_{i\in S_0^c}t_i\Big)^r=1,\qquad
\Big(\prod_{i\in S_0^c}t_i\Big)^s=(-1)^{r-1}.
\end{align}
Their unique simultaneous solution is
$\prod_{i\in S_0^c}t_i=(-1)^{r-1}$.
Indeed, this sign has the required $r$-th power, and its $s$-th power is
$(-1)^{r-1}$ because $(r-1)(s-1)$ is even.
Consequently, the monic polynomial with roots $t_i$, $i\in S_0^c$, has
constant term $-1$, and can be written uniquely as
\begin{align}\label{eq:b-from-roots}
\prod_{i\in S_0^c}(t-t_i)
=t^r-\sum_{j=1}^{r-1}b_jt^j-1.
\end{align}
This defines $b_1,\ldots,b_{r-1}$, and hence the matrices $U',D',T$.

For every $i\in S_0^c$, Equation~\eqref{eq:companion-root-powers} implies
\begin{align}
\kappa t_i^s(\tilde Y t_i^r-Q)=Q\tilde Y t_i^r-Q.
\end{align}
Hence all the distinct roots in Equation~\eqref{eq:b-from-roots} are
roots of
\begin{align}
\kappa\tilde Y t^{r+s}-\kappa Q t^s-Q\tilde Y t^r+Q.
\end{align}
The characteristic polynomial of $T$ divides this polynomial, hence we obtain 
Equation~\eqref{eq:augmentation-companion}.

Finally, construct $E$ as in Equation~\eqref{eq:E} and set
\begin{align}\label{eq:reconstructed-aij}
a_{ij}=\begin{cases}
(E_{ii}/E_{jj})b_{j-i},&i<j,\\
-(E_{ii}/E_{jj})b_{r-i+j},&i>j.
\end{cases}
\end{align}
By the equivalence proved above, these entries satisfy both
augmentation equations~\eqref{eq:aug-two-matrix-equations}.
This gives the required reconstruction for every generic point of $F=0$.

\subsubsection{Concluding remarks} We have proved that every augmentation point in the torus $\tilde{X}\tilde{Y}\not=0$ projects to $F=0$, and that every generic point of $F=0$ is the projection of an augmentation point.
Since $F$ is irreducible and is divisible by neither $\tilde X$ nor $\tilde Y$, these generic points in the torus are dense in its affine zero set.
Thus the closure of the augmentation curve in the affine plane is exactly $F=0$.
This proves Proposition~\ref{prop:augmentation-projection}
and Theorem~\ref{thm:FF}.
Finally, to be pedantic, note that the coefficient of
$\tilde X^{\frac1{r+s}\binom{r+s}{s}}\tilde Y^{\binom{r+s}{s}}$
in $F$ is one, so $F$ has the same normalization as $\tilde F$.
Consequently, $\tilde F=F$, as asserted.

\section{Topological recursion} \label{sec: topological recursion}

\subsection{Meromorphy on the augmentation curve}
Fix $r$ and $s$ such that $\gcd(r,s)=1$, and let $C$ be the curve defined over $\mathbb C(Q^{1/2})$ by $F(\tilde X,\tilde Y)=0$, where $F$ is given in~\eqref{eq:F}. 

\begin{proposition}\label{prop:augmentation-correlators}
For $g\geq0$, $n\geq1$, such that $2g-2+n>0$, the $n$-differential
$\omega^{(g)}_n(\tilde X_1,\dots,\tilde X_n)$ given by \eqref{eq: averaging} extends to a symmetric meromorphic
$n$-differential on $C^n$.
\end{proposition}

In order to prove this proposition, we apply a general algebraic argument that allows us to pass from the auxiliary curve $\mathbb{C}\mathrm{P}^1$ with $X$ and $y$ given on it in terms of a global rational coordinate $z$ by \eqref{eq:X} and \eqref{eq:y}, respectively, to $C$ given by the vanishing locus of $\eqref{eq:F}$, using the averaging given by \eqref{eq: averaging} and the corresponding statement on $\omega^{\mathsf{ext},(g)}_n$. The unstable cases $(g,n)=(0,1)$ and $(0,2)$ are considered below separately. The general algebraic argument is captured by the following lemma:

\begin{lemma}\label{lem:algebraic lift}
Let $f(z)$ be a rational function on $\mathbb{C}\mathrm{P}^1$ with coefficients in
$\mathbb C(Q^{1/2})$, regular at $z=0$. Consider the following expression defined in terms of its local expansion at $z=0$ 
\begin{equation}\label{eq:average}
\frac1s\sum_{a=0}^{s-1}f(z(\xi^aX)).
\end{equation}
It is a formal power series in $\tilde{X} = X^s$. This formal power series 
extends to a meromorphic function on $C$.
\end{lemma}

\begin{proof} It is a version of the Lagrange interpolation formula. Let
\begin{equation}\label{eq:interpolation}
\widehat{f}(\tilde X,\tilde Y)\coloneqq 
\sum_{|S|=s}\left(\frac1s \sum_{a\in S}f(z_a)\right)
\prod_{\substack{|T|=s\\T\ne S}}
\frac{\tilde Y-c_T}{c_S-c_T}.
\end{equation}
It is a polynomial in $\tilde{Y}$ whose coefficients are symmetric
rational expressions in the roots $z_1,\ldots,z_{r+s}$ of $P(z)$ given in~\eqref{eq:z-root-polynomial}. Thus the coefficients of $\tilde{Y}$ in $\widehat{f}$ belong to $\mathbb{C}(Q^{1/2})(\tilde X)$. Hence $\widehat{f}$ restricts to a meromorphic function on $C$. 

Recall that on the curve $C$ given by $F(\tilde{X},\tilde{Y})=0$ we have $\tilde Y=c_S$ and on its distinguished branch near $\tilde{X}=0$, $\tilde{Y} = c_{S_0}$. Hence
\begin{align}
    \widehat{f}|_{\tilde Y=c_S} =\widehat{f}(\tilde X,c_S)= \frac 1s \sum_{a\in S} f(z_a)
\end{align}
which specializes to \eqref{eq:average} on the distinguished branch near $\tilde{X}=0$.

In the proof of Theorem~\ref{thm:corrected-algebraicity} we identified $z_a$, $a\in S_0$, with the branches $z(\xi^iX)$, $i=0,\dots,s-1$, at $X=0$. Therefore, the restriction of $\widehat f$ to the distinguished branch agrees with the formal power series \eqref{eq:average}. Since $\widehat f$ is globally meromorphic on $C$, this proves the claim.
\end{proof}

\begin{proof}[Proof of Proposition~\ref{prop:augmentation-correlators}] Notice that $dX/X$ on the auxiliary curve is equal to $s^{-1}d\tilde{X}/\tilde{X}$ and thus extends to a meromorphic differential on $C$. Then use that $\omega^{\mathsf{ext},(g)}_n $ is a meromorphic $n$-differential and apply Lemma~\ref{lem:algebraic lift} to $[\hbar^{2g-2+n}]W_n^{\mathsf{ext}} = \omega^{\mathsf{ext},(g)}_n \prod_{i=1}^n X_i/dX_i$ in each of its $n$ variables. 
\end{proof}

In particular, this implies that for $2g-2+n>0$, $\omega^{(g)}_n(\tilde X_1,\dots,\tilde X_n)$ extends on $C^n$ as an analytic continuation of 
\begin{align}
    \sum_{a_1,\dots,a_n=0}^{s-1} \omega^{\mathsf{ext},(g)}_n(\xi^{a_1}X_1,\dots,\xi^{a_n}X_n)
\end{align}
(the factor of $s^{-n}$ is now hidden in the lift of $\prod_{i=1}^ndX_i/X_i$). 
In other words, let the point $(q_1,\dots,q_n)\in C^n$ be given by the choice of branches $c_{S_1},\dots,c_{S_n}$ over $\tilde{X}_1,\dots,\tilde{X}_n$ in~\eqref{eq:F} (notation: $q_i\sim c_{S_i}$). Then \emph{locally} we have 
\begin{align} \label{eq: pull-back}
    \omega^{(g)}_n|_{(q_1,\dots,q_n)} = (\tilde{X}^*)^{\otimes n} \sum_{\substack{a_i\in S_i\\ i=1,\dots,n}} \omega^{\mathsf{ext},(g)}_n|_{(z_{a_1}(\tilde X_1),\dots,z_{a_n}(\tilde X_n))}.
\end{align}
In this formula, we consider differentials in $z_a$ as $dz_a = \frac{dz_a(\tilde{X})}{d\tilde{X}} d\tilde{X}$. By $\tilde{X}^*$ we denote the pullback of a differential under the map $\tilde{X}\colon C \to \mathbb{C}\mathrm{P}^1$.

Now let us discuss the unstable cases. The case $(g,n)=(0,1)$ follows directly from the construction by analytic continuation. We indeed have directly from the definition
\begin{align}
\omega^{(0)}_1|_q = \Big(\tilde y\,\frac{d\tilde X}{\tilde X}\Big)|_q = \log(c_S)\,\frac{d\tilde X}{\tilde X} = \frac1s\sum_{a\in S}y(z_a(\tilde X))\,\frac{d\tilde X}{\tilde X} = \tilde X^*\sum_{a\in S}\Big(y\,\frac{dX}{X}\Big)|_{z_a(\tilde X)},
\end{align}
where the point $q$ is given over $\tilde{X}$ by the choice of branch $c_S$. However, these are locally defined logarithmic differentials, rather than global meromorphic differentials.

In the case $(g,n)=(0,2)$ we have (cf.~\cite[Theorem 9.3]{DPSS-HOMFLY}):
\begin{align}
    \left(\omega^{(0)}_2 - \frac{d\tilde{X}_1 d\tilde{X}_2}{(\tilde{X}_1-\tilde{X}_2)^2}\right) \bigg|_{q_1,q_2} & = \tilde{X_1}^* \tilde{X_2}^* \sum_{\substack{ a_1\in S_1 \\ a_2\in S_2}} \Big(\frac{dz(\tilde{X}_1) dz(\tilde X_2)}{(z(\tilde{X}_1)-z(\tilde{X}_2))^2} - \frac{dX_1dX_2}{(X_1-X_2)^2} \Big)|_{z_{a_1}(\tilde X_1),z_{a_2}(\tilde X_2)}
    \\ \notag 
    & = \sum_{\substack{ a_1\in S_1 \\ a_2\in S_2}} \frac{dz_{a_1}(\tilde X_1) dz_{a_2}(\tilde X_2)}{(z_{a_1}(\tilde X_1)-z_{a_2}(\tilde X_2))^2} - s \frac{d\tilde{X}_1 d\tilde{X}_2}{(\tilde{X}_1-\tilde{X}_2)^2}.
\end{align}
Here $q_i\sim c_{S_i}$, $i=1,2$, with the same branchwise interpretation as above.

\subsection{Poles on the augmentation curve} First, we discuss the stable range of $(g,n)$.

\begin{proposition} Let $g\geq 0$, $n\geq 1$, $2g-2+n>0$. The only poles of $\omega^{(g)}_n$, considered as a meromorphic $n$-differential on $C^n$, are at the branching points of $\tilde{X}$ over $V_1$ and $V_2$. 
\end{proposition}

\begin{proof} Indeed, by Equation~\eqref{eq: pull-back}, $\omega^{(g)}_{n}$ is regular in one of its arguments at each point $q$ if $\omega^{\mathsf{ext},(g)}_{n}$ is regular at every point $z_a(\tilde X)$, $a=1,\dots,r+s$. Thus, the possible poles are indeed only over $\tilde{X}_i = V_1$ or $V_2$. 

For simplicity of exposition, let us focus on the case $\omega^{(g)}_1$ (the general case is exactly the same, one just has to suppress all variables but one in the notation). The ramification profile of $\tilde X:C\to \mathbb{CP}^1$ over $V_1,V_2$ was already established in Prop.~\ref{prop:ramificationprofile}.
Consider the branches $z_1(\tilde X),\dots,z_{r+s}(\tilde X)$ near $\tilde X = V_1$ and assume that $z_1(\tilde X)=z_2(\tilde X)$ at $\tilde X = V_1$. Let $q\in C$ correspond to the choice of branch $c_S$ over $\tilde X$. There are three possible cases:
\begin{itemize}
    \item $1,2\notin S$. Then the multiplicity of the corresponding point $q$ in the fiber over $\tilde X = V_1$ is equal to $1$. Notice that the sum over $a\in S$ in  Equation~\eqref{eq: pull-back} contains only regular terms. Hence, $\omega^{(g)}_1$ is regular at $q$ as well. This corresponds to the $\binom{r+s-2}{s}$ points over $V_1$ with multiplicity $1$ in Prop.~\ref{prop:ramificationprofile}.
    \item $1,2\in S$. Then the multiplicity of the corresponding point $q$ in the fiber over $\tilde X = V_1$ is equal to $1$. Notice that the sum over $a\in S$ in  Equation~\eqref{eq: pull-back} contains two singular terms. However, the sum of these two singular terms is regular by the linear loop equations for $\omega^{\mathsf{ext},(g)}_1$. This corresponds to the $\binom{r+s-2}{s-2}$ points over $V_1$ with multiplicity $1$ in Prop.~\ref{prop:ramificationprofile}.
    \item $1\in S$, $2\notin S$, or vice versa. The multiplicity of the corresponding point $q$ in the fiber over $\tilde X = V_1$ is equal to $2$. Notice that there is exactly one summand in Equation~\eqref{eq: pull-back} that contains a singular term, and it generates the Laurent part of $\omega^{(g)}_1$ at $q$. This corresponds to the $\binom{r+s-2}{s-1}$ branch points over $V_1$ with multiplicity $2$ in Prop.~\ref{prop:ramificationprofile}.
\end{itemize}
\end{proof}

\begin{remark}
This argument also implies that the principal parts at different
ramification points over $V_1$ (resp., $V_2$) coincide in suitable
local coordinates, with the other arguments fixed.
Indeed, these principal parts are obtained by pulling back the same
auxiliary principal part along the local isomorphisms to the germ
at $p_1$ (resp., $p_2$), given by the selected colliding root.
\end{remark}

The pole structure of $\omega^{(0)}_2$ is more complicated. Note that it is not the standard Bergman kernel associated to the curve $C$. From its formula at $q_1\sim c_{S_1}$, $q_2\sim c_{S_2}$
\begin{align}\label{omega02calibrated}
        \omega^{(0)}_2|_{q_1,q_2} = \sum_{\substack{ a_1\in S_1 \\ a_2\in S_2}} \frac{dz_{a_1}(\tilde X_1) dz_{a_2}(\tilde X_2)}{(z_{a_1}(\tilde X_1)-z_{a_2}(\tilde X_2))^2} - (s-1) \frac{d\tilde{X}_1 d\tilde{X}_2}{(\tilde{X}_1-\tilde{X}_2)^2}
\end{align}
we see that its possible poles occur along the shifted diagonals $\tilde{X}(q_1) = \tilde X(q_2)$. At a generic point of such a diagonal, the coefficient of the double pole is $|S_1\cap S_2|+1-s$. In particular if $|S_1\cap S_2|=s-1$,  the double pole cancels  and $\omega^{(0)}_2$ is regular at a generic point of the corresponding shifted diagonal. On the main diagonal, where 
$S_1=S_2=S$, the expected bi-residue $1$ occurs. The bi-differential \eqref{omega02calibrated} appeared in \cite{GJKS}, and was called the \emph{calibrated annulus kernel}. 

The pole structure of $\omega^{(0)}_1$ is irrelevant for the local
recursion; what matters here is that $d\tilde y$ does not vanish
at the ramification points over $V_1$ and $V_2$.

\subsection{Loop equations and projection formula} The arguments of~\cite{BEO} applied to the extended differentials $\{\omega^{\mathsf{ext},(g)}_n\}$, imply the following

\begin{proposition} At all branching points over $V_1$ and $V_2$ (equivalently, at all poles of $\omega^{(g)}_n$ for $2g-2+n>0$) the system of differentials $\{\omega^{(g)}_n\}$ satisfies the linear and quadratic loop equations. 
\end{proposition}

\begin{proof} Consider a point $o\in C$ over, say, $V_1$, where the branches $c_{S'\sqcup \{1\}}$ and $c_{S'\sqcup \{2\}}$ collide, $S'\subseteq \{3,\dots,r+s\}$, $|S'|=s-1$. Note that the deck transformation descends from the deck transformation on the auxiliary curve that acts on the roots of $P$ by interchanging $z_1$ and $z_2$ near the only branching point $u$ over $V_1$ there.

Use Equation~\eqref{eq: pull-back} and its respective versions for the non-stable correlators. Let $q_i\sim c_{S'\sqcup \{i\}}$, $i=1,2$. Then for the linear loop equations, we have the following (again, for simplicity, we present the argument for $\omega^{(g)}_1$, otherwise one has to suppress the non-relevant variables in the notation):
\begin{align} \label{eq:linearloop}
    \omega^{(g)}_1|_{q_1} + \omega^{(g)}_1|_{q_2} = \tilde X^* \left(\omega^{\mathsf{ext},(g)}_1|_{z_{1}(\tilde X)}+\omega^{\mathsf{ext},(g)}_1|_{z_{2}(\tilde X)} +2\sum_{a\in S'} \omega^{\mathsf{ext},(g)}_1|_{z_{a}(\tilde X)} \right).
\end{align}
Here $\omega^{\mathsf{ext},(g)}_1|_{z_{1}(\tilde X)}+\omega^{\mathsf{ext},(g)}_1|_{z_{2}(\tilde X)}$ is regular at $u$ by the linear loop equations for the extended differentials~\eqref{eq:linear-loop-ext}, the last summand is regular at $u$ by the structure of the poles of the extended differentials. Thus, the right hand side of~\eqref{eq:linearloop} is a pull-back of a regular differential with respect to $\tilde{X}^*$, thus it is a regular differential vanishing at $o\in C$. 

The computation for the quadratic loop equations is very similar. We have:
\begin{align}
    & \omega^{(g-1)}_2|_{q_1,q_2} + \sum_{g_1+g_2=g} \omega^{(g_1)}_{1}|_{q_1} \omega^{(g_2)}_1|_{q_2} = 
    \\ \notag &
    \tilde X_1^* \tilde X_2^* \bigg(\omega^{\mathsf{ext},(g-1)}_2|_{z_{1}(\tilde X),z_{2}(\tilde X)} + \sum_{g_1+g_2=g} 
    \omega^{\mathsf{ext},(g_1)}_1|_{z_{1}(\tilde X)}\, \omega^{\mathsf{ext},(g_2)}_1|_{z_{2}(\tilde X)}
    \\ \notag & \qquad \qquad 
    + \sum_{a\in S'} \Big(\omega^{\mathsf{ext},(g-1)}_2|_{z_{1}(\tilde X),z_{a}(\tilde X)} +\omega^{\mathsf{ext},(g-1)}_2|_{z_{a}(\tilde X),z_{2}(\tilde X)} \Big) 
    \\ \notag & \qquad \qquad 
    + \sum_{a\in S'} \sum_{g_1+g_2=g}
    \Big(\omega^{\mathsf{ext},(g_1)}_1|_{z_{1}(\tilde X)}\, \omega^{\mathsf{ext},(g_2)}_1|_{z_{a}(\tilde X)} + 
    \omega^{\mathsf{ext},(g_1)}_1|_{z_{a}(\tilde X)}\, \omega^{\mathsf{ext},(g_2)}_1|_{z_{2}(\tilde X)} \Big)
    \\ \notag & \qquad \qquad 
    + \sum_{a,b\in S'} \Big(\omega^{\mathsf{ext},(g-1)}_2|_{z_{a}(\tilde X),z_{b}(\tilde X)} + \sum_{g_1+g_2=g} 
    \omega^{\mathsf{ext},(g_1)}_1|_{z_{a}(\tilde X)}\, \omega^{\mathsf{ext},(g_2)}_1|_{z_{b}(\tilde X)} \Big) \bigg).
\end{align}
Here in the first line on the right-hand side the expression is regular at $u$ by the quadratic loop equations~\eqref{eq:quadratic-loop-ext} for the extended differentials, the second and the third line by the linear loop equations~\eqref{eq:linear-loop-ext}, and the fourth line by the structure of poles of $\{\omega^{\mathsf{ext},(g)}_n\}$. Thus, the expression on the right hand side is a pull-back of a regular differential, and $\tilde X_1^*\tilde X_2^*$ generates a double zero at $o\in C$.
\end{proof}

\subsection{Projection formula and multiplicities}

The projection property for $\omega^{(g)}_n$, $g\geq 0$, $n\geq 1$, $2g-2+n>0$, is different from the standard one. Typically, one uses the vanishing of the $A$-periods in combination with the given germs of the Laurent parts of the differentials at their poles to uniquely extend the differentials globally via Cauchy formula. 

In this case, we use a different kind of uniqueness:
\begin{lemma}\label{lem:uniqueness} There exist at most one meromorphic differential $\omega$ on $C$ with the only poles at the branching point over $V_1$ and $V_2$, with the prescribed Laurent parts there, and locally at $q\sim c_S$ compatible with the pull-back formula
$\omega|_q=\tilde X^*\sum_{a\in S}\omega^{\mathsf{ext}}|_{z_a(\tilde X)}$
for some meromorphic differential $\omega^{\mathsf{ext}}$ on the auxiliary curve.
\end{lemma}

\begin{proof} Indeed, the difference of two meromorphic differentials satisfying these properties is a holomorphic differential compatible with the pull-back formula. But any holomorphic differential locally compatible with the pull-back formula vanishes since there are no non-vanishing holomorphic differentials on the auxiliary curve.
\end{proof}
\begin{remark}
The stable differentials $\omega_n^{(g)}$ are globally \emph{exact}
in each variable. More precisely, there exists a meromorphic
function $H_n^{(g)}$ on $C^n$ such that
\begin{align}
\omega_n^{(g)}=d_1\cdots d_n H_n^{(g)},\qquad 2g-2+n>0. 
\end{align}
Indeed, the stable auxiliary differentials $\omega_n^{\mathsf{ext},(g)}$ are rational and
have vanishing residues in each variable, hence admit rational
multiprimitives on $(\mathbb{CP}^1)^n$. Applying the algebraic
lift of Lemma~\ref{lem:algebraic lift} in each variable gives
global meromorphic multiprimitives on $C^n$, and differentiation
commutes with the sheet-sum formula \eqref{eq: pull-back}.
Consequently, all periods of the stable differentials vanish
in each variable. This is very different from ordinary topological recursion on higher genus spectral curves, where the generated differentials are in general not globally exact in each variable.
\end{remark}
Now, consider the projection formula 
		\begin{align} \label{eq:projection-omega}
			\omega^{(g)}_{n}(\xi_1,\dots,\xi_n) = \frac{1}{\binom{s+r-2}{s-1}^n} \sum_{r_1,\dots,r_n} \Big( \prod_{i=1}^n \Res_{\xi'_i\to r_i}	\int_{r_i}^{\xi'_i} \omega^{(0)}_2(\bullet,\xi_i) \Big) \omega^{(g)}_{n}(\xi'_1,\dots,\xi'_n),
		\end{align}
where $\xi_i$ are some local coordinates and the sum is taken over $r_1,\dots,r_n$ running through the branch points of $\tilde X$ over $V_1$ and $V_2$,
with $\omega^{(0)}_2$ used instead of the standard Bergman kernel. A simple observation that the right hand side of this projection formula indeed satisfies the conditions of Lemma~\ref{lem:uniqueness} implies that Equation~\eqref{eq:projection-omega} holds for all $g\geq 0$, $n\geq 1$, $2g-2+n>0$, and thus we conclude with the following 

\begin{theorem} \label{thm:TRonAugmentation}
The system of differentials $\{\omega^{(g)}_n\}$ encoding the colored HOMFLY-PT polynomials of torus $(r,s)$ knot satisfies topological recursion on the corresponding augmentation curve.
\end{theorem}

\begin{remark}
    In this form, with exactly the same caveat on non-standard replacement for the Bergman kernel, this theorem is conjectured in~\cite{GJKS} and its statement is supported in \emph{op.~cit.} both by computations and physical reasoning. 
\end{remark}

\begin{remark} One can compute directly the multiplicities of polar parts in the projection formula~\eqref{eq:projection-omega}. At two ramification points over $V_1$ (the argument for $V_2$ is literally the same),
indexed by subsets $S'_1,S'_2\subseteq\{3,\dots,r+s\}$, $|S'_1|=|S'_2|=s-1$, choose compatible
coordinates $\tilde X(u_1)=V_1+u_1^2$, $\tilde X(u_2)=V_1+u_2^2$. Then along the corresponding shifted diagonal
\begin{align} \label{eq:02-shifted}
\omega^{(0)}_2(u_1,u_2)
=
\frac{du_1du_2}{(u_1-u_2)^2}
+\bigl(|S'_1\cap S'_2|-s+1\bigr)
\frac{4u_1u_2\,du_1\,du_2}{(u_1^2-u_2^2)^2}
+O(1).
\end{align}
The first term comes from the $\frac{dz_1(\tilde X_1)dz_2(\tilde X_2)} {(z_1(\tilde X_1)-z_2(\tilde X_2))^2}$ summand in~\eqref{omega02calibrated}; the second term comes from the remaining terms representing $\frac{dz_a(\tilde X_1)dz_a(\tilde X_2)} {(z_a(\tilde X_1)-z_a(\tilde X_2))^2}$ for common $a\in S_1'\cap S_2'$ and the $\frac{d\tilde X_1 d\tilde X_2}{(\tilde X_1-\tilde X_2)^2}$ term. All other
pairs of roots in~\eqref{omega02calibrated} give holomorphic local contributions. 

By the linear loop equation, the principal part of $\omega^{(g)}_1$ (again, for clarity of exposition we restrict ourselves to this case) written in the local coordinate $u_1$ is odd. The primitive of the second summand in~\eqref{eq:02-shifted} is even in $u_1$, thus its product with an odd principal part doesn't contribute to the residue in~\eqref{eq:projection-omega}. The residue with the primitive of the first summand in~\eqref{eq:02-shifted}
reproduces the common principal part at every ramification point. Since we take the sum over $\binom{s+r-2}{s-1}$ ramification points over $V_1$, we get at each such point the required principal part with the multiplicity $\binom{s+r-2}{s-1}$, and we divide by this multiplicity in the resulting projection formula~\eqref{eq:projection-omega}.

Similar local computation shows that the right hand side of Equation~\eqref{eq:projection-omega} is regular at the non-ramified preimages of $V_1$ and $V_2$. 
\end{remark}

\subsection{Symmetry under the exchange of \texorpdfstring{$r$ and $s$}{r and s}}

An important symmetry of torus-knot invariants is their
invariance under the exchange of $r$ and $s$. In particular,
with the same framing convention for both presentations,
the colored HOMFLY-PT polynomials satisfy
\begin{align}
H_\lambda^{r,s}(Q^{-1/2},e^{\hbar/2})
=
H_\lambda^{s,r}(Q^{-1/2},e^{\hbar/2}),
\end{align}
although this symmetry is not manifest in the Rosso--Jones
formula \eqref{eq:Rosso-Jones}.

Because of the presentation-dependent prefactor in
\eqref{eq:OV}, the corresponding Ooguri--Vafa partition
functions are related by a rescaling of the power-sum
variables. More precisely, writing $Z^{(r,s)}$ for the
partition function defined in \eqref{eq:OV}, we have
\begin{equation}\label{eq:OV-rs-symmetry}
Z^{(r,s)}(p_1,p_2,\ldots)
=
Z^{(s,r)}\left(
Q^{(r-s)/2}p_1,\,
Q^{r-s}p_2,\,
Q^{3(r-s)/2}p_3,\ldots
\right).
\end{equation}
Indeed, this follows from the homogeneity identity
\begin{align}
s_\lambda(Q^{(r-s)/2}p_1,Q^{r-s}p_2,\ldots)
=
Q^{|\lambda|(r-s)/2}s_\lambda(p).
\end{align}

From this, we can conclude how the correlation differentials
are related by the swap of $r$ and $s$. Namely,
\begin{align}
    \frac{\partial^n \log Z^{(r,s)}}
    {\partial p_{m_1}\cdots\partial p_{m_n}}
    \bigg\vert_{p=0}
    =
    Q^{(m_1+\cdots+m_n)(r-s)/2}
    \frac{\partial^n \log Z^{(s,r)}}
    {\partial p_{m_1}\cdots\partial p_{m_n}}
    \bigg\vert_{p=0}.
\end{align}
Hence the local germs already satisfy
\begin{align}
    \omega_{n}^{(g),(r,s)}
    (\tilde{X}_1,\ldots,\tilde{X}_n)
    =
    \omega_{n}^{(g),(s,r)}
    (Q^{(r-s)/2}\tilde{X}_1,\ldots,
     Q^{(r-s)/2}\tilde{X}_n),
\end{align}
where the right-hand side includes the differential pullbacks.

From the construction of Sec.~\ref{sec:TrueCurve}, one can
introduce a Möbius transformation in \eqref{eq:tildeX}
of the form
\begin{align}
    M(z)=\frac{Q^{1/2}-z}{1-Q^{1/2}z},
    \qquad
    \tilde X_{s,r}(M(z))
    =Q^{(r-s)/2}\tilde X_{r,s}(z),
\end{align}
now with $r$ roots approaching $0$. The same construction
of the augmentation curve holds with $r,s$ interchanged,
where the selected subset $S$ is replaced by the complement
$S^c$ under the Möbius transformation. The differentials
$\omega_n^{(g)}$ are then constructed from the $r$-roots
averaging.

The augmentation curve is invariant under the $r$-$s$
swap up to an overall factor. We have
\begin{align}
    F_{r,s}(\tilde X,\tilde Y)
    =
    Q^{-\frac{r-s}{2(r+s)}\binom{r+s}{s}}
    F_{s,r}(Q^{(r-s)/2}\tilde X,\tilde Y).
\end{align}
This transforms the leading term by
$\tilde X\to Q^{(r-s)/2}\tilde X$ and relates the
constant terms according to
\begin{align}
    (Q^{1/2})^{\frac{s}{r+s}\binom{r+s}{s}}
    =
    Q^{-\frac{r-s}{2(r+s)}\binom{r+s}{s}}
    (Q^{1/2})^{\frac{r}{r+s}\binom{r+s}{s}},
\end{align}
as desired. See \cite{Ng-Introduction} for the
corresponding invariance of knot contact homology.

These properties extend globally on $C^n$ to the
correlation differentials. In particular, we show that
the bidifferentials $\omega^{(0),(r,s)}_2$
and $\omega^{(0),(s,r)}_2$ can be identified globally.

\begin{lemma}
    The bidifferentials $\omega^{(0),(r,s)}_2$ and
    $\omega^{(0),(s,r)}_2$ agree globally under the
    identification $\tilde X\mapsto
    Q^{(r-s)/2}\tilde X$, $\tilde Y\mapsto\tilde Y$.
\end{lemma}

\begin{proof}
    In the following, we compare the bidifferentials
    sheetwise. Since the argument is independent of the
    chosen sheets, they coincide globally.

    Let us first observe that a Möbius transformation $M$
    leaves the Cauchy bidifferential invariant:
    \begin{align}
        \frac{dM(z)dM(w)}{(M(z)-M(w))^2}
        =
        \frac{dz\,dw}{(z-w)^2}.
    \end{align}
    The Möbius transformation used to identify the
    auxiliary coordinates therefore does not affect
    the bidifferentials.

    Let $S_1,S_2\subseteq\{1,\ldots,r+s\}$ be the
    subsets corresponding to the sheets
    $q_1\sim c_{S_1}$ and $q_2\sim c_{S_2}$, with
    $|S_1|=|S_2|=s$. On these sheets, we have
    \begin{align}\label{om02rs}
        \omega^{(0),(r,s)}_2\big|_{q_1,q_2}
        =
        \sum_{\substack{i\in S_1\\j\in S_2}}
        \frac{dz_i(\tilde X_1)dz_j(\tilde X_2)}
        {(z_i(\tilde X_1)-z_j(\tilde X_2))^2}
        \quad -(s-1)
        \frac{d\tilde X_1d\tilde X_2}
        {(\tilde X_1-\tilde X_2)^2}.
    \end{align}
    This will be identified with the $r$-$s$ swapped
    bidifferential on the $r$-element complementary
    subsets $S_1^c,S_2^c$:
    \begin{align}\label{om02sr}
        \omega^{(0),(s,r)}_2\big|_{q_1,q_2}
        =
        \sum_{\substack{i\in S_1^c\\j\in S_2^c}}
        \frac{dz_i(\tilde X_1)dz_j(\tilde X_2)}
        {(z_i(\tilde X_1)-z_j(\tilde X_2))^2} -(r-1)
        \frac{d\tilde X_1d\tilde X_2}
        {(\tilde X_1-\tilde X_2)^2}.
    \end{align}
    Here \eqref{om02sr} is understood after the coordinate
    rescaling and Möbius transformation. Note that the
    sum over $i,j$ runs over the complementary subsets
    of $\{1,\ldots,r+s\}$ compared to \eqref{om02rs},
    which identifies the corresponding sheets.

    Let us now compute the full sum over all roots $z_i$.
    For each $j$, we have
    \begin{align}\label{eq:full-Cauchy-rs}
        \sum_{i=1}^{r+s}
        \frac{dz_i(\tilde X_1)dz_j(\tilde X_2)}
        {(z_i(\tilde X_1)-z_j(\tilde X_2))^2}
        =
        \frac{d\tilde X_1d\tilde X_2}
        {(\tilde X_1-\tilde X_2)^2}.
    \end{align}
    Indeed, by the definition of $P(z;\tilde X)$ in \eqref{eq:z-root-polynomial},
    \begin{align}
        \prod_{i=1}^{r+s}
        (z_i(\tilde X_1)-z_j(\tilde X_2))
        =
        \mathrm{const}\cdot
        P(z_j(\tilde X_2);\tilde X_1)
        =
        \mathrm{const}\cdot
        (\tilde X_2-\tilde X_1)
        (1-Q^{1/2}z_j(\tilde X_2))^r.
    \end{align}
    Taking $d_{\tilde X_1}d_{\tilde X_2}\log$
    gives \eqref{eq:full-Cauchy-rs}. The analogous
    identity holds for the sum over $j$.

    Consequently,
    \begin{align}
        \sum_{i,j=1}^{r+s}
        \frac{dz_i(\tilde X_1)dz_j(\tilde X_2)}
        {(z_i(\tilde X_1)-z_j(\tilde X_2))^2}
        =
        (r+s)
        \frac{d\tilde X_1d\tilde X_2}
        {(\tilde X_1-\tilde X_2)^2}.
    \end{align}
    On the other hand, the full sum can be split as
    \begin{align}
        \sum_{i,j=1}^{r+s}
        \frac{dz_i(\tilde X_1)dz_j(\tilde X_2)}
        {(z_i(\tilde X_1)-z_j(\tilde X_2))^2}
        &=
        \left(
        \sum_{\substack{1\leq i\leq r+s\\j\in S_2}}
        +
        \sum_{\substack{i\in S_1\\1\leq j\leq r+s}}
        +
        \sum_{\substack{i\in S_1^c\\j\in S_2^c}}
        -
        \sum_{\substack{i\in S_1\\j\in S_2}}
        \right)
        \frac{dz_i(\tilde X_1)dz_j(\tilde X_2)}
        {(z_i(\tilde X_1)-z_j(\tilde X_2))^2}\\ \notag
        &=
        2s\frac{d\tilde X_1d\tilde X_2}
        {(\tilde X_1-\tilde X_2)^2}
        +
        \left(
        \sum_{\substack{i\in S_1^c\\j\in S_2^c}}
        -
        \sum_{\substack{i\in S_1\\j\in S_2}}
        \right)
        \frac{dz_i(\tilde X_1)dz_j(\tilde X_2)}
        {(z_i(\tilde X_1)-z_j(\tilde X_2))^2}.
    \end{align}
    Comparing with the previous sum, we find
    \begin{align}
        \sum_{\substack{i\in S_1^c\\j\in S_2^c}}
        \frac{dz_i(\tilde X_1)dz_j(\tilde X_2)}
        {(z_i(\tilde X_1)-z_j(\tilde X_2))^2}
         =
        (r-s)
        \frac{d\tilde X_1d\tilde X_2}
        {(\tilde X_1-\tilde X_2)^2}
        +
        \sum_{\substack{i\in S_1\\j\in S_2}}
        \frac{dz_i(\tilde X_1)dz_j(\tilde X_2)}
        {(z_i(\tilde X_1)-z_j(\tilde X_2))^2}.
    \end{align}
    Substituting into \eqref{om02sr} gives
    \begin{align}
        \omega^{(0),(r,s)}_2\big|_{q_1,q_2}
        =
        \omega^{(0),(s,r)}_2\big|_{q_1,q_2},
    \end{align}
    where the identification of coordinates is understood.
    The equality holds on every generic pair of sheets
    and hence extends globally.
\end{proof}

We conclude:
\begin{corollary}
    Under the $r$-$s$ swap, the correlation differentials
    are identical on the augmentation curve after the
    coordinate rescaling. More precisely,
    \begin{align}
        \omega_n^{(g),(r,s)}
        (\tilde X_1,\ldots,\tilde X_n)
        =
        \omega_n^{(g),(s,r)}
        (Q^{(r-s)/2}\tilde X_1,\ldots,
         Q^{(r-s)/2}\tilde X_n)
    \end{align}
    holds globally on $C^n$.
\end{corollary}

\begin{proof}
    The equality of the local germs follows from
    \eqref{eq:OV-rs-symmetry}. Since the stable
    differentials extend meromorphically to the
    augmentation curves, this identity extends globally.
    The disk differentials are identified under the same
    coordinate rescaling, and the preceding lemma
    establishes the global identification of the bidifferentials.
\end{proof}

Importantly, the differentials of the auxiliary curve
$\omega_n^{\mathsf{ext},(g)}$ do not transform by a
simple rescaling under the $r$-$s$ swap, whereas the
differentials on the true physical$=$augmentation curve
transform naturally under this exchange. The remaining
rescaling can be absorbed into the choice of coordinates
used in the recursion. It is therefore a matter
of convention.

\subsection{A topological recursion proof of the Aganagic--Vafa
conjecture for torus knots}

A very important application of topological recursion is the
computation of Gromov--Witten invariants of toric Calabi--Yau
threefolds from their mirror curves. This is known as the
Remodeling Conjecture \cite{Bouchard:2007ys}, which states
that ordinary topological recursion (with the standard
Bergman kernel) computes the topological string amplitudes
of toric Calabi--Yau threefolds. This conjecture was proved
in \cite{EynardOrantinBKMP,FangLiuZongBKMP}.

An important example is the resolved conifold, which acts
as a fundamental building block in the topological vertex
construction. The mirror curve of the resolved conifold
is also the starting point for the auxiliary curve,
which was obtained by a fractional $SL_2$-transformation
by Brini, Eynard, and Mari\~no \cite{BEM11}.

The augmentation curve considered here is not a standard
toric mirror curve. However, Aganagic and Vafa \cite{AV12}
proposed a relation as a knot-dependent mirror curve
from the resolved conifold through the large-$N$ transition.
In particular, they conjectured that the closed-string
amplitudes associated with these knot-dependent curves
coincide with those of the resolved conifold.
This conjecture concerns arbitrary knots in $S^3$.

It remained an open question whether topological recursion
on the augmentation curve (thus the knot dependent mirror curve understood as the physical curve) computes the corresponding
topological string amplitudes. In \cite{GJKS}, a modified
topological recursion for torus knots was proposed,
using a calibrated annulus kernel. The corresponding
topological recursion statement is proved above.

We now study the associated free energies, denoted by
$\omega_0^{(g)}$. To derive them, we first show that
the dilaton equation can be transferred from the auxiliary
curve to the augmentation curve.

\begin{proposition}\label{prop:dilaton}
    For $2g-2+n>0$, the following dilaton equation holds
    on the augmentation curve:
    \begin{align}
        &(2-2g-n)\,
        \omega_n^{(g)}(\xi_1,\ldots,\xi_n)
        =
        \frac{1}{\binom{r+s-2}{s-1}}
        \sum_{\rho}
        \operatorname*{Res}_{\xi'\to\rho}
        \left(
            \int_{\rho}^{\xi'}\omega_1^{(0)}
        \right)
        \omega_{n+1}^{(g)}
        (\xi',\xi_1,\ldots,\xi_n),
    \end{align}
    where $\xi_i$ are some local coordinates and the sum runs over all ramification points
    $\rho$ of $\tilde X$ over $V_1$ and $V_2$.

    In particular, for $g>1$, we define the free energy
    by
    \begin{align}
        \omega_0^{(g)}
        =
        \frac{1}{(2-2g)\binom{r+s-2}{s-1}}
        \sum_{\rho}
        \operatorname*{Res}_{\xi'\to\rho}
        \left(
            \int_{\rho}^{\xi'}\omega_1^{(0)}
        \right)
        \omega_1^{(g)}(\xi').
    \end{align}
    With the same normalization of the free energies,
    the free energy on the augmentation curve coincides
    with the free energy on the auxiliary curve:
    \begin{align}
        \omega_0^{(g)}
        =
        \omega_0^{\mathsf{ext},(g)}.
    \end{align}
\end{proposition}

\begin{proof}
    Consider an auxiliary ramification point $u$ over $V_1$,
    where the roots $z_1(\tilde X)$ and $z_2(\tilde X)$
    collide. As in the proof of the loop equations, the
    corresponding ramification points $\rho$ on $C$ are
    indexed by subsets
    $S'\subseteq\{3,\ldots,r+s\}$, $|S'|=s-1$,
    with local sheets $S'\sqcup\{1\}$ and $S'\sqcup\{2\}$.

    Let $q\sim c_{S'\sqcup\{1\}}$ and
    $\xi_i\sim c_{S_i}$. By the pullback formula
    \eqref{eq: pull-back}, we have
    \begin{align}
        \omega_1^{(0)}\big|_q
        &=
        \tilde X^*
        \left(
            \omega_1^{\mathsf{ext},(0)}
            \big|_{z_1(\tilde X)}
            +
            \sum_{b\in S'}
            \omega_1^{\mathsf{ext},(0)}
            \big|_{z_b(\tilde X)}
        \right),\\ \notag 
        \omega_{n+1}^{(g)}
        \big|_{q,\xi_1,\ldots,\xi_n}
        &=
        \tilde X^*\tilde X^*_1...\tilde X^*_n
        \sum_{\substack{a_i\in S_i\\i=1,\ldots,n}}
        \left(
            \omega_{n+1}^{\mathsf{ext},(g)}
            \big|_{z_1(\tilde X),z_{a_1},\ldots,z_{a_n}}
            +
            \sum_{b\in S'}
            \omega_{n+1}^{\mathsf{ext},(g)}
            \big|_{z_b(\tilde X),z_{a_1},\ldots,z_{a_n}}
        \right).
    \end{align}

    Choose a local coordinate $t$ at $\rho$ such that
    $\tilde X=V_1+t^2$. All spectator differentials are
    regular in $\tilde X$, so their primitives are
    holomorphic functions of $t^2$ vanishing at $t=0$.
    By the linear loop equations for the extended
    differentials, the polar part of the active stable
    differential has an even Laurent coefficient in $t$.
    Consequently, every term containing a spectator
    contribution has zero residue. Thus
    \begin{align}
        &\Res_{q\to\rho}
        \left(\int_\rho^q\omega_1^{(0)}\right)
        \omega_{n+1}^{(g)}(q,\xi_1,\ldots,\xi_n)
        =
        \tilde X^*_1...\tilde X^*_n
        \sum_{\substack{a_i\in S_i\\i=1,\ldots,n}}
        \Res_{z\to u}
        \left(\int_u^z\omega_1^{\mathsf{ext},(0)}\right)
        \omega_{n+1}^{\mathsf{ext},(g)}
        (z,z_{a_1},\ldots,z_{a_n}).
    \end{align}
    The right-hand side is independent of $S'$.
    Since there are exactly
    $\binom{r+s-2}{s-1}$ such ramification points over
    each $V_j$, summing over both $V_1$ and $V_2$
    and applying the auxiliary dilaton equation gives
    \begin{align}
        \frac{1}{\binom{r+s-2}{s-1}}
        \sum_{\rho}\Res_{q\to\rho}
        \left(\int_\rho^q\omega_1^{(0)}\right)
        \omega_{n+1}^{(g)}(q,\xi_1,\ldots,\xi_n)
        &=
        (2-2g-n)
        \tilde X^*_1...\tilde X^*_n
        \sum_{\substack{a_i\in S_i\\i=1,\ldots,n}}
        \omega_n^{\mathsf{ext},(g)}
        (z_{a_1},\ldots,z_{a_n})\\ \notag 
        & =
        (2-2g-n)\omega_n^{(g)}
        (\xi_1,\ldots,\xi_n).
    \end{align}
    The argument over $V_2$ is identical.

    Finally, applying the same residue comparison to
    $\omega_1^{(g)}$ for $g>1$ yields
    \begin{align}
        \omega_0^{(g)}
        =
        \frac{1}{2-2g}
        \sum_{u\in\{p_1,p_2\}}
        \Res_{z\to u}
        \left(\int_u^z\omega_1^{\mathsf{ext},(0)}\right)
        \omega_1^{\mathsf{ext},(g)}(z)
        =
        \omega_0^{\mathsf{ext},(g)}.
    \end{align}
\end{proof}

The closed-string free energies of the resolved conifold
are well known. They were computed in physics
\cite{GV98Conifold} and mathematically using
Gromov--Witten theory \cite{FaberPandharipande2000}.
Related results have been rederived and generalized to
further spectral curves via $x$--$y$ duality in
topological recursion
\cite{Hock:2022wer,Hock:2022pbw,Alexandrov2025XYSwap}
and its extension to logarithmic topological recursion
\cite{Hock-qDiLog,ABDKS-LogTR} in
\cite{Hock:2025wlm,Banerjee:2025qgx,Hock:2026qrp}. In particular, those articles prove independence of the framing, thus independence of $r,s$ under the $SL_2$ transformation used in \cite{BEM11}.


Thus, using the known expression for the auxiliary free
energy, we obtain for $g>1$
\begin{align}\label{freeenergy}
    \omega_{0}^{\mathsf{ext},(g)}
    =
    \frac{B_{2g}B_{2g-2}}
    {2g(2g-2)(2g-2)!}
    +
    \frac{B_{2g}\mathrm{Li}_{3-2g}(Q)}
    {2g(2g-2)!}.
\end{align}
Here $B_k$ are the Bernoulli numbers, defined by $t/(e^t-1)=\sum_{k\geq0}B_k t^k/k!$, and $\mathrm{Li}_{3-2g}(Q)$ denotes the rational function whose
Taylor expansion at $Q=0$ is $\sum_{d\geq1}d^{2g-3}Q^d$. Combining this with Proposition~\ref{prop:dilaton},
we conclude:

\begin{corollary}\label{cor:freeenergy}
    The free energy $\omega_0^{(g)}$ on the torus-knot
    augmentation curve is given by \eqref{freeenergy}.
    In particular, it coincides with the free energy
    of the resolved conifold.
\end{corollary}

\begin{remark}
    Corollary~\ref{cor:freeenergy} proves the
    topological recursion version of the Aganagic--Vafa
    conjecture for all torus knots stated in \cite{AV12}.
\end{remark}

\printbibliography

\end{document}